\documentclass[11pt,onecolumn]{IEEEtran}
\usepackage{amsthm}
\usepackage{times,amssymb,amsmath,amsfonts,float,nicefrac,color,bbm,mathrsfs,float}
\usepackage{algorithm,enumerate,multirow,caption,tikz,graphicx}
\usepackage{algpseudocode}
\usepackage[noadjust]{cite}
\usepackage{subcaption}
\usepackage[top=1in,bottom=1in,left=1in,right=1in]{geometry}
\allowdisplaybreaks

\newcommand{\RN}[1]{%
  \textup{\expandafter{\romannumeral#1}}%
}

\usetikzlibrary{shapes.misc, positioning}
\usetikzlibrary{decorations.pathreplacing}

\tikzset{
  block/.style    = {draw, thick, rectangle, minimum width = 3em},
  sblock/.style      = {draw, thick, rectangle, minimum height = 3em,
    minimum width = 3em}, 
}

\tikzset{XOR/.style={draw,circle,append after command={
        [shorten >=\pgflinewidth, shorten <=\pgflinewidth,]
        (\tikzlastnode.north) edge (\tikzlastnode.south)
        (\tikzlastnode.east) edge (\tikzlastnode.west)
        }
    }
}

\newcommand\remove[1]{}

\allowdisplaybreaks

\newtheorem{theorem}{Theorem}    
\newtheorem{definition}{Definition}
\newtheorem{proposition}{Proposition}

\newtheorem{lemma}{Lemma}
\newtheorem{corollary}{Corollary}

\newtheorem{example}{Example}
\newtheorem{remark}{Remark}

\newtheorem{cnstr}{Construction}

\def\mathbi#1{{\textbf{\textit #1}}}

\newcommand{\cA}{\mathcal{A}}
\newcommand{\cB}{\mathcal{B}}

\newcommand{\cG}{\mathcal{G}}

\newcommand{\cR}{\mathcal{R}}
\newcommand{\cS}{\mathcal{S}}
\newcommand{\cT}{\mathcal{T}}

\newcommand{\cY}{\mathcal{Y}}

\DeclareMathOperator{\odd}{odd}

\DeclareMathOperator{\even}{even}

\DeclareMathOperator{\poly}{poly}

\DeclareMathOperator{\avg}{avg}
\DeclareMathOperator{\BSC}{BSC}
\DeclareMathOperator{\BEC}{BEC}

\begin{document}
\title{Reed-Muller codes polarize}

\author{Emmanuel Abbe \and \hspace*{.4in} Min Ye}

\maketitle
{\renewcommand{\thefootnote}{}\footnotetext{

\vspace{-.2in}
 
\noindent\rule{1.5in}{.4pt}

E. Abbe is with the Mathematics Institute and the School of Computer and Communication Sciences at EPFL, Switzerland, and the Program in Applied and Computational Mathematics and the Department of Electrical Engineering in Princeton University, USA. M. Ye is with Department of Electrical Engineering, Princeton University, Princeton, NJ, email: yeemmi@gmail.com. 
}

\renewcommand{\thefootnote}{\arabic{footnote}}
\setcounter{footnote}{0}

\begin{abstract}
Reed-Muller (RM) codes and polar codes are generated by the same matrix $G_m= \bigl[\begin{smallmatrix}1 & 0 \\ 1 & 1 \\ \end{smallmatrix}\bigr]^{\otimes m}$ but using different subset of rows.  RM codes select simply rows having largest weights. Polar codes select instead rows having the largest conditional mutual information proceeding top to down in $G_m$; while this is a more elaborate and channel-dependent rule, the top-to-down ordering has the advantage of making the conditional mutual information polarize, giving directly a capacity-achieving code on any binary memoryless symmetric channel (BMSC). RM codes are yet to be proved to have such property.

In this paper, we reconnect RM codes to polarization theory. It is shown that proceeding in the RM code ordering, i.e., not top-to-down but from the lightest to the heaviest rows in $G_m$, the conditional mutual information again polarizes. We further demonstrate that it does so faster than for polar codes. This implies that $G_m$ contains another code, different than the polar code and called here the twin code, that is provably capacity-achieving on any BMSC. This proves a necessary condition for RM codes to achieve capacity on BMSCs. It further gives a sufficient condition if the rows with largest conditional mutual information correspond to the heaviest rows, i.e., if the twin code is the RM code. We show here that the two codes bare similarity with each other and give further evidence that they are likely the same. 
\end{abstract}

\section{Introduction}\label{sect:intro}
Reed-Muller codes have long been conjectured to be capacity-achieving.\footnote{See \cite{Kudekar17} for accounts on this conjecture.} This was recently settled in \cite{Kudekar16STOC,Kudekar17} for the special case of the binary erasure channel (BEC), and in \cite{Abbe15,Sberlo18} for special cases of extremal rates on both the BEC and the binary symmetric channel (BSC). The general conjecture of achieving capacity on the BSC and more generally any binary memoryless symmetric (BMS)  channel\footnote{Recall that a BMS channel is a channel $W:\{0,1\}\to\cY$ such that there is a permutation $\pi$ on the output alphabet $\cY$ satisfying i) $\pi^{-1}=\pi$ and ii) $W(y|1)=W(\pi(y)|0)$ for all $y\in\cY$.} at constant rate remains open to date.

The research activity on RM codes has resurged in part due to the development of polar codes \cite{Arikan09,ArikanTelatar}. Both RM codes and polar codes are generated by selecting subset of rows from the same base matrix $G_m=    \bigl[\begin{smallmatrix}
      1 & 0 \\
      1 & 1 \\
   \end{smallmatrix}\bigr]^{\otimes m}$. Polar codes select the rows by tracking the conditional mutual information of each row given the past rows when proceeding top to down in $G_m$ (see Section \ref{sect:bp} for precise definitions). In this specific ordering, Ar\i kan was able to show a polarization result \cite{Arikan09}, i.e., that most of the rows have a conditional mutual information that tend to either 0 or 1. This in turn implies fairly directly that the code resulting from keeping the high conditional mutual information rows is capacity-achieving on any BMS channel.

A first drawback of polar codes is that the code construction, i.e., identifying the rows having high conditional mutual information, is non-trivial. In particular, there is to date no known explicit characterization of the row selection except for the BEC. This is however not an algorithmic limitation as there are known efficient algorithms that approximate arbitrarily closely the values of the conditional mutual information for each row \cite{Tal13}. Two more important drawbacks are that polar codes are not universal \cite{Hassani09}, as their row selection is channel dependent, and their scaling law is sub-optimal compared to that of random codes \cite{Hassani14} or RM codes \cite{Hassani18}, making their error probability at short block length not as competitive as could be \cite{Mondelli14}. On the flip side, polar codes benefit from a powerful analytical framework, the polarization framework \cite{Arikan09,Guruswami15,Blasiok18}, which allows to give performance guarantees, and from an efficient successive decoding algorithm. Their performance at short block length has also been improved with the addition of outer codes and list decoding algorithms \cite{Tal15}. With these attributes, polar codes are in  position to enter the 5G standards \cite{3gpp}.

On the other side, RM codes benefit from a  simple and universal code construction: selecting the heaviest rows is trivial and depends only on the capacity of the channel and not the actual channel. Further, it is already known that RM codes would have an optimal scaling law {\it if} they were proved to be capacity-achieving  \cite{Hassani18}. Performance improvements over polar codes at short block length were also demonstrated in \cite{Mondelli14}. On the flip side, the main challenges of RM codes are (i) their analytical framework, with the difficulty of obtaining performance guarantees, (ii) the absence of an efficient decoding algorithm that succeeds up to capacity for the constant rate regime.

\subsection{Recent progress}
As mentioned earlier, progress has recently been made on both points (i) and (ii). We mention briefly here a few references for decoding algorithms \cite{Reed54,Dumer04,Saptharishi17,Santi18}, as this not the main focus of this paper. We refer to our parallel paper with a new decoding algorithm \cite{YA18} for a more detailed discussion of those. 

We now discuss performance guarantees. In \cite{Kudekar17}, the case of the BEC is settled by exploiting results on the threshold of monotone Boolean functions, benefiting from the fact that the events of decoding failures for erasures correspond to monotone properties of Boolean functions. With this link, general results from Boolean function analysis \cite{Kahn88,Talagrand94,Friedgut96} come to rescue and allow to close the conjecture for the BEC. While this gives an elegant proof, it has the downside of relying on a ``Hammer'' result \cite{Kahn88,Bourgain92} that does not seem to generalize easily beyond erasures due to the loss of the monotonicity property. The approach of \cite{Abbe15} relies instead on the polynomial characterization of RM codes (whose codewords can be viewed as the evaluation of bounded degree multivariate Boolean polynomials) and on the weight enumerator of RM codes. A downside of that approach is that it is currently not reaching the constant rate regime; although some recent progress towards that goal was  made in \cite{Sberlo18}.

Moreover, none of the above seem to shed light on the connection between polar and RM codes, which remains a recurrent question. A first attempt to connect RM codes to polar code was made in \cite{Wigderson16}, using the double conditional rank measure in relation to the algebraic view of polarization \cite{Yuval15}; conjectures based on this approach were left in \cite{Wigderson16}.

\subsection{This paper}
Considering the developments so far, it may appear that the simplicity of the RM code construction fires back in the complexity of their analysis, in contrast to polar codes, where a more elaborate construction allows to benefit from the powerful polarization framework.

The goal of this paper is to show that this is not a necessary limitation, and that RM codes benefit too from a polarization phenomenon, slightly different but potentially more effective than that of polar codes.  
We view RM codes as the evaluation of multivariate polynomials and make use of the recursive Plotkin construction\footnote{Any $d$-degree polynomial can be decomposed with two $(d-1)$-degree polynomials as $f(x^d)=x_d f_1(x^{d-1}) + f_0(x^{d-1})$.}  $(\mathbi{u},\mathbi{u}+\mathbi{v})$  
\cite{Macwilliams77}, which is similar in nature to the recursive construction of polar codes. 
Together with the establishment of an ordering on the conditional mutual information of RM codes, we derive a new polarization result for the RM code ordering, and obtain consequently the capacity-achieving result for the resulting code (called the twin code) that selects high conditional mutual information rows in the RM code ordering. The proof relies solely on classical polarization properties and some algebra. 

This result gives in particular a necessary condition for RM codes to achieve capacity on any BMS channel. It gives also a sufficient  condition if the  rows  with  largest  conditional mutual information correspond to the heaviest rows, i.e., if the twin code is the RM code. We give a relaxed version of the latter, showing that the twin code is similar to the RM code, and give further evidence that it is in fact the RM code. Note that in the contrary case, i.e., if the twin code were not equivalent to the RM code, then RM codes would not achieve capacity on some BMS channels.

\section{Background}
\subsection{RM codes}

Let us consider the polynomial ring $\mathbb{F}_2[Z_1,Z_2,\dots,Z_m]$ of $m$ variables over $\mathbb{F}_2$. Since $Z^2=Z$ in $\mathbb{F}_2$, the following set of $2^m$ monomials forms a basis of $\mathbb{F}_2[Z_1,Z_2,\dots,Z_m]$:
$$
\{\prod_{i\in A}Z_i: A\subseteq [m] \}, \text{~where~} \prod_{i\in \emptyset}Z_i :=1.
$$
Next we associate every subset $A\subseteq [m]$ with a row vector $\mathbi{v}_m(A)$ of length $2^m$, whose components are indexed by binary vectors $\mathbi{z}=(z_1,z_2,\dots,z_m) \in \{0,1\}^m$.
The vector $\mathbi{v}_m(A)$ is defined as follows:
\begin{equation}\label{eq:gg}
\mathbi{v}_m(A,\mathbi{z}) = \prod_{i\in A} z_i,
\end{equation}
where $\mathbi{v}_m(A,\mathbi{z})$ is the component of $\mathbi{v}_m(A)$ indexed by $\mathbi{z}$,
i.e., $\mathbi{v}_m(A,\mathbi{z})$ is the evaluation of the monomial $\prod_{i\in A}Z_i$ at $\mathbi{z}$.
For $0\le r \le m$, the set of vectors 
$$
\{\mathbi{v}_m(A):A\subseteq[m],|A|\le r\}
$$
forms a basis of the $r$-th order Reed-Muller code $\cR(m,r)$ of length $n:=2^m$
and dimension $\sum_{i=0}^r \binom{m}{i}$.

\begin{definition}
 The $r$-th order Reed-Muller code $\cR(m,r)$ code is defined as the following set of binary vectors 
$$
\cR(m,r) := \left\{\sum_{A\subseteq[m],|A|\le r}u(A) \mathbi{v}_m(A): u(A)\in\{0,1\} 
\text{~~for all~} A\subseteq[m],|A|\le r\right\}.
$$
\end{definition}

\begin{example}\label{ex:m3}
We write out a basis of $\cR(3,3)$ as follows:
$$
\begin{array}{ccccccccc}
 (z_1,z_2,z_3) & (1,1,1) & (1,1,0) & (1,0,1) & (1,0,0) & (0,1,1) & (0,1,0) & (0,0,1) & (0,0,0) \\
A=\{3,2,1\} & 1 & 0 & 0 & 0 & 0 & 0 & 0 & 0 \\
A=\{2,1\} & 1 & 1 & 0 & 0 & 0 & 0 & 0 & 0 \\
A=\{3,1\} & 1 & 0 & 1 & 0 & 0 & 0 & 0 & 0 \\
A=\{3,2\} & 1 & 0 & 0 & 0 & 1 & 0 & 0 & 0 \\
A=\{1\} & 1 & 1 & 1 & 1 & 0 & 0 & 0 & 0 \\
A=\{2\} & 1 & 1 & 0 & 0 & 1 & 1 & 0 & 0 \\
A=\{3\} & 1 & 0 & 1 & 0 & 1 & 0 & 1 & 0 \\
A=\emptyset & 1 & 1 & 1 & 1 & 1 & 1 & 1 & 1 
\end{array},
$$
where the first row lists the index $\mathbi{z}$ of each component, and the second to the last rows are $\mathbi{v}_3(A), A\subseteq[3]$.
\end{example}

In this paper, we prove a polarization result for Reed-Muller codes. To that end,
we define a total order on all the subsets of $[m]$ as follows: 
\begin{definition} [total order] \label{def:tto}
For $A=\{a_1,a_2,\dots,a_{|A|}\},B=\{b_1,b_2,\dots,b_{|B|}\}\subseteq[m]$, where $a_1>a_2>\dots>a_{|A|}$ and $b_1>b_2>\dots>b_{|B|}$,
 we write $A<B$ if either of the following two conditions is satisfied:
\begin{enumerate}
\item $|A|>|B|$;
\item $|A|=|B|$, and there is an integer $i\in\{1,2,\dots,|A|\}$ such that $a_j=b_j\forall j<i$ and $a_i<b_i$.
\end{enumerate}
\end{definition}
It is easy to check that for any two sets $A,B\subseteq[m]$, one of the following three relations must hold: $A<B,A=B$ or $A>B$. Therefore, this is indeed a total order on all the subsets of $[m]$.
Note that condition $1)$ ensures that picking the `largest' sets layer by layer gives the RM code. Condition $2)$ says how to order the rows within a layer (e.g., if the code dimension requires breaking a layer), but any ordering resulting from a permutation of the elements in $[m]$ would be equivalent. We pick this convention as we like to see the $m$-th element as the `new element' when running the forthcoming inductions.

For $m=3$, the rows in Example~\ref{ex:m3} are listed in the increasing order of the set $A$.
Let $(U_A^{(m)}:A\subseteq[m])$ be $n:=2^m$ i.i.d. Bernoulli-$1/2$ random variables.
We use the shorthand notation $U_{<A}^{(m)}:=(U_{A'}^{(m)}:A'\subseteq[m],A'<A)$ and $U^{(m)}:=(U_A^{(m)}:A\subseteq[m])$.
Next we define another $n$ i.i.d. Bernoulli-$1/2$ random variables $X_{\mathbi{z}}^{(m)},\mathbi{z}\in\{0,1\}^m$ by
$$
(X_{\mathbi{z}}^{(m)},\mathbi{z}\in\{0,1\}^m) :=
\sum_{A\subseteq[m]} U_A^{(m)} \mathbi{v}_m(A).
$$
We transmit $X_{\mathbi{z}}^{(m)},\mathbi{z}\in\{0,1\}^m$ through $n$ independent copies of a BMS channel $W:\{0,1\}\to\cY$, and we denote the corresponding channel outputs as $Y_{\mathbi{z}}^{(m,W)},\mathbi{z}\in\{0,1\}^m$. 
Let 
$$
X^{(m)}:=(X_{\mathbi{z}}^{(m)}:\mathbi{z}\in\{0,1\}^m) \quad \text{~and~} \quad
Y^{(m,W)}:=(Y_{\mathbi{z}}^{(m,W)}:\mathbi{z}\in\{0,1\}^m).
$$
Since $W$ is symmetric and $(X_{\mathbi{z}}^{(m)},\mathbi{z}\in\{0,1\}^m)$ are also i.i.d.\ Bernoulli-$1/2$ random variables, we have for all $\mathbi{z}\in\{0,1\}^m$, 
  $H(X_{\mathbi{z}}^{(m)}|Y_{\mathbi{z}}^{(m,W)})=1-I(W)$, and therefore
$H(U^{(m)}|Y^{(m,W)})=H(X^{(m)}|Y^{(m,W)})=nH(X_{\mathbi{z}}^{(m)}|Y_{\mathbi{z}}^{(m,W)})=n(1-I(W))$, where
$H(\cdot|\cdot)$ is conditional entropy and $I(\cdot)$ is the channel capacity (or the symmetric capacity for channels that are not BMS).
Thus 
$$
\sum_{A\subseteq[m]} H(U_A^{(m)}|Y^{(m,W)},U_{<A}^{(m)}) = n(1-I(W)).
$$
For convenience, we use the notation
\begin{equation}\label{eq:ha}
H_A^{(m,W)}:=H(U_A^{(m)}|Y^{(m,W)},U_{<A}^{(m)}).
\end{equation}
From now on, we omit to specify $W$ from the notation $H_A^{(m,W)}, Y^{(m,W)}$ and $Y_{\mathbi{z}}^{(m,W)}$ when the underlying channel is not important, i.e., we write them as $H_A^{(m)}, Y^{(m)}$ and $Y_{\mathbi{z}}^{(m)}$.
Therefore,
\begin{equation}\label{eq:kw}
\sum_{A\subseteq[m]} H_A^{(m)} = n(1-I(W)).
\end{equation}
We also define the channel $W_A^{(m)}$ as the binary-input channel that takes $U_A^{(m)}$ as input and $Y^{(m)},U_{<A}^{(m)}$ as outputs, i.e., $W_A^{(m)}$ is the channel seen by the successive decoder when decoding $U_A^{(m)}$. 

In order to state our main results, we also need the definition of the Bhattacharyya parameter.
Let $(X,Y)$ be a pair of random variables such that $X$ has Bernoulli-$1/2$ distribution and $Y$ takes values from a finite alphabet $\cY$. The Bhattacharyya parameter is defined as 
$$
Z(X|Y):= \sum_{y\in\cY} \sqrt{P_{Y|X}(y|0) P_{Y|X}(y|1)}.
$$
Similarly to $H_A^{(m)}$, for a subset $A\subseteq[m]$ and a BMS channel $W$ we use the shorthand notation
$$
Z_A^{(m)}=Z_A^{(m,W)}:=Z(U_A^{(m)}|Y^{(m,W)},U_{<A}^{(m)}).
$$

\subsection{Polarization}  \label{sect:bp}
The polar coding transform is given by the following $n\times n$ matrix
\begin{equation}\label{eq:tbd}
G_m:=\left[\begin{array}{cc} 1 & 0 \\ 1 & 1 \end{array} \right]^{\otimes m},
\end{equation}
where $\otimes$ is the Kronecker product and $n=2^m$.

\begin{figure}[h]
\centering
\begin{subfigure}{.45\linewidth}
\centering
\begin{tikzpicture}
\draw
 node at (0,10.5) [] (u1)  {$U_1$}
 node at (0,9) [] (u2)  {$U_2$}
 node at (1.5,10.5) [XOR,scale=1.2] (x1) {}
  node at (2.5,10.5) [] (xx1)  {$X_1$}
 node at (2.5,9) [] (xx2)  {$X_2$}
 node at (3.8,10.5) [block] (v1)  {$W$}
 node at (3.8,9) [block] (v2)  {$W$}
 node at (5.5,10.5) [] (y1)  {$Y_1$}
 node at (5.5,9) [] (y2)  {$Y_2$};
 \draw[very thick,->](u1) -- node {}(x1);
 \draw[very thick,->](u2) -| node {}(x1);
 \draw[very thick,->](x1) -- (xx1);
 \draw[very thick,->](u2) -- (xx2);
 \draw[very thick,->](xx1) -- (v1);
 \draw[very thick,->](xx2) -- (v2);
 \draw[very thick,->](v1) -- node {}(y1);
 \draw[very thick,->](v2) -- node {}(y2);
\end{tikzpicture}
\caption{Apply the polar matrix $G_1$ to i.i.d. uniform random variables $(U_1,U_2)$, and then transmit the results through two copies of $W$. Under successive decoder, this transforms two copies of $W$ into $W^-:U_1\to Y_1,Y_2$ and $W^+:U_2\to U_1,Y_1,Y_2$.}
\end{subfigure}
\hspace*{0.2in}
\begin{subfigure}{.45\linewidth}
\centering
\begin{tikzpicture}
\draw
 node at (0,1.5) [] (u1)  {$W$}
 node at (0,0) [] (u2)  {$W$}
 node at (1.2,1.5) [] (v1)  {}
 node at (1.2,0) [] (v2)  {}
 node at (3.7,1.5) [] (x1)  {}
 node at (3.7,0) [] (x2)  {}
 node at (3.7,1.5) [] (x1)  {}
 node at (3.7,0) [] (x2)  {}
 node at (5.2,1.5) [] (y1)  {$W^-$}
 node at (5.2,0) [] (y2)  {$W^+$}
 node at (2.45,0.75) [text width=2cm,align=center] {Polar\\Transform};
 \draw[thick] (1.2,-0.5) rectangle (3.7,2);
 \draw[very thick,->](u1) -- node {}(v1);
 \draw[very thick,->](u2) -- node {}(v2);
 \draw[very thick,->](x1) -- node[above] {``$-$"}(y1);
 \draw[very thick,->](x2) -- node [above] {``$+$"} (y2);
\end{tikzpicture}
\caption{We take two independent copies of $W$ as inputs. After the transform, we obtain a ``worse" channel $W^-:U_1\to Y_1,Y_2$ and a ``better" channel $W^+:U_2\to U_1,Y_1,Y_2$.}
\end{subfigure}
\caption{Polar Transform}
\label{fig:ptb}
\end{figure}
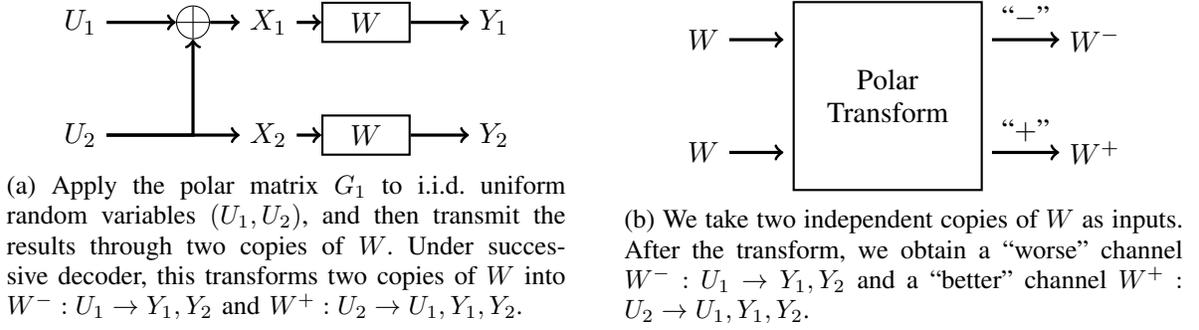

The basic idea of polar coding is that after applying the polar matrix $G_1$ defined in \eqref{eq:tbd}, we obtain a ``worse" channel $W^-:U_1\to Y_1,Y_2$ and a ``better" channel $W^+:U_2\to U_1,Y_1,Y_2$; see Fig.~\ref{fig:ptb} for an illustration. 
This statement can be quantified by the relations among the conditional entropy:
\begin{align}
H(U_1|Y_1,Y_2) \ge H(X_1|Y_1) \ge H(U_2|U_1,Y_1,Y_2), \label{eq:kbd} \\
H(U_1|Y_1,Y_2) + H(U_2|U_1,Y_1,Y_2) = 2 H(X_1|Y_1).  \label{eq:ndm}
\end{align}
Similar relations among the Bhattacharyya parameters were also proved in \cite[Proposition 5]{Arikan09}:
\begin{align}
Z(U_2|Y_1,Y_2,U_1) &= (Z(X_1|Y_1))^2, \label{eq:ks}\\
 Z(U_1|Y_1,Y_2)  &\ge Z(X_1|Y_1) . \label{eq:lsl}
\end{align}

Moreover, if $H(X_1|Y_1)$ is bounded away from $0$ and $1$, then the gap between $H(U_1|Y_1,Y_2)$ and $H(X_1|Y_1)$ is bounded away from $0$, and so does the gap between $H(X_1|Y_1)$ and $H(U_2|U_1,Y_1,Y_2)$. (By \eqref{eq:ndm}, these two gaps are the same.) In other words, if $W$ is neither noiseless nor completely noisy, then $W^-$ is strictly worse than $W$, and $W^+$ is strictly better. The rigorous statement is as follows.
\begin{lemma}[\cite{Arikan09}] \label{lm:sas}
Let $(X_1,Y_1)$ and $(X_2,Y_2)$ be an independent pair of random variables, where $X_1$ and $X_2$ take values in $\{0,1\}$. 
For all $\epsilon>0$, there is $\delta(\epsilon)>0$ such that 
$$
H(X_1|Y_1),H(X_2|Y_2)\in(\epsilon,1-\epsilon)
$$
implies
$$
H(X_1+X_2|Y_1,Y_2)\ge \max(H(X_1|Y_1),H(X_2|Y_2)) + \delta(\epsilon).
$$
\end{lemma}

Polar coding scheme consists of applying the polar matrix $G_m$ to $n$ i.i.d. uniform random variables and transmitting the results through $n$ copies of $W$. This amounts to iteratively applying the ``$+$" and ``$-$" polar transforms to $W$, and almost all the resulting bit-channels seen by the successive decoder becomes either noiseless or completely noisy.
Let $\{\widetilde{H}_A^{(m)}\}_{A\subseteq[m]}$ (resp., $\{\widetilde{Z}_A^{(m)}\}_{A\subseteq[m]}$) be the conditional entropy (resp., Bhattacharyya parameter) of each row given all the past rows when decoding top to down in $G_m$.

\begin{theorem}[Polarization of polar codes \cite{Arikan09}]
For every BMS channel $W$, almost all elements in the set $\{\widetilde{H}_A^{(m)}\}_{A\subseteq[m]}$ are close to either $0$ or $1$ when $m$ is large. More precisely,
for any $0<\epsilon<1/10$ and any $\delta_n>\exp(-n^{0.499})$, there is a constant $M(\epsilon,\delta_n)$ such that for every $m>M(\epsilon,\delta_n)$,
$$
\frac{\left| \left\{A\subseteq[m]:  \widetilde{H}_A^{(m)} > 1-\epsilon \right\}
\cup \left\{A\subseteq[m]:  \widetilde{Z}_A^{(m)} < \delta_n \right\}
\right|}{2^m}
\ge 1- o(1).
$$
\end{theorem}

\section{Main results} \label{sect:mm}

Our main results are summarized in the following theorems.

\begin{theorem}[Polarization of RM codes]\label{thm:m3}
For every BMS channel $W$, almost all elements in the set $\{H_A^{(m)}\}_{A\subseteq[m]}$ are close to either $0$ or $1$ when $m$ is large. More precisely,
for any $0<\epsilon<1/10$, any $\delta_n=\poly(1/n)$ and any $0<\gamma<1/2$, there is a constant $M(\epsilon,\delta_n,\gamma)$ such that for every $m>M(\epsilon,\delta_n,\gamma)$,
$$
\frac{\left| \left\{A\subseteq[m]:  H_A^{(m)} > 1-\epsilon \right\}
\cup \left\{A\subseteq[m]:  Z_A^{(m)} < \delta_n \right\}
\right|}{2^m}
\ge 1- m^{\gamma-1/2}.
$$
\end{theorem}
As mentioned above, the basis vectors of RM codes $\{\mathbi{v}_m(A):A\subseteq[m]\}$ are exactly the row vectors of the polar matrix $G_m$ defined in \eqref{eq:tbd}.
However, these rows are arranged in different orders for RM codes and polar codes, which makes the polarization of RM codes fundamentally different from that of polar codes.

As an immediate consequence of Theorem~\ref{thm:m3}, we can construct a family of capacity-achieving codes.
\begin{theorem}[Twin codes] \label{thm:cac}
For a BMS channel $W$ and $\delta_n=\poly(1/n)$, let 
$$
\cG(m, \delta_n):= \left\{A\subseteq[m]:  Z_A^{(m)} < \delta_n \right\}
$$
and define the family of twin codes from the codewords
$$
\cT(m,\delta_n):=
\left\{\sum_{A\in \cG}u(A) \mathbi{v}_m(A): u(A)\in\{0,1\} 
\text{~~for all~} A\in \cG(m,\delta_n) \right\},
$$
where $\mathbi{v}_m(A)$ is defined in \eqref{eq:gg}.
Then for $\delta_n=o(n^{-2})$, $\cT(m,\delta_n)$ achieves the capacity of $W$ under successive decoding.
\end{theorem}
This theorem tells us that we can construct capacity achieving codes using successive decoder under the RM ordering (i.e., ordered by weights). Note that none of the above give algorithmic results.

To establish the above, we need the following notion of ordering between the different conditional entropies in the RM ordering.
\begin{definition} [Partial order]
For $A=\{a_1,a_2,\dots,a_{|A|}\},B=\{b_1,b_2,\dots,b_{|B|}\} \subseteq [m]$, $A \ne B$, where $a_1<a_2<\dots<a_{|A|}$ and $b_1<b_2<\dots<b_{|B|}$, we define
\begin{align}
A \prec B \text{ if and only if } |A| \ge |B| \text{ and } a_i \le b_i, \, \forall i \le |B|.     
\end{align}
\end{definition}
The reason why we set the above to be $A \prec B$ and not $A \succ B$ is that this gives an order with $[m]$ as the `first' set and $\emptyset$ as the `last' set, which corresponds to the first and last sets in the RM code ordering.

\begin{theorem}\label{order2}
If $A \prec B$, then $H_A^{(m)} \ge H_B^{(m)}$.
\end{theorem}

According to Theorem~\ref{thm:cac} and Theorem \ref{order2}, the twin code $\cT(m,\delta_n)$ tend to select sets $A$ with small cardinality, which is similar to RM codes (that exactly selects sets with the smallest cardinality). However, we do not establish here whether this is exactly the RM code or not. We do give a positive indication by establishing that this is exactly the RM code up to $n=16$ for the BSC, and leave the general case for future work; see Section~\ref{sect:tisr} for details.

\section{Proof outline} \label{sect:pst}

In order to explain the main ideas of the proof, we introduce the following definition. 
\begin{definition}[Increasing chain of sets]
Let $A_0=\emptyset$ and $A_m=[m]$. We say that $A_0 \subseteq A_1 \subseteq A_2 \subseteq \dots \subseteq A_m$ is an increasing chain of sets if $|A_i|=i$ for all $i=0,1,2,\dots,m$.
\end{definition}

A main step in our argument consist in proving the following two theorems:
\begin{theorem}\label{thm:m1}
For every BMS channel $W$, every $m>0$ and every increasing chain of sets $\emptyset=A_0 \subseteq A_1 \subseteq A_2 \subseteq \dots \subseteq A_m=[m]$, we have
$$
H_{A_0}^{(m)} \le H_{A_1}^{(m)} \le H_{A_2}^{(m)} \le \dots \le H_{A_m}^{(m)}.
$$
\end{theorem}

\begin{theorem}\label{thm:m2}
For every BMS channel $W$ and every $\epsilon>0$, there is a constant $D(\epsilon)$ (which is independent of $m$ and $W$) such that for every $m>0$ and every increasing chain of sets $\emptyset=A_0 \subseteq A_1 \subseteq A_2 \subseteq \dots \subseteq A_m=[m]$,
$$
\left|\left\{ i\in\{0,1,\dots,m\}:\epsilon < H_{A_i}^{(m)} < 1-\epsilon \right\} \right| \le D(\epsilon). 
$$
\end{theorem}

\begin{figure}
\centering
\begin{tikzpicture}
\draw (0,0) -- (15,0);
\draw (0,1) -- (15,1);
\draw [fill] (0.2,0) circle [radius=1.5pt];
\draw [fill] (1,0) circle [radius=1.5pt];
\draw [fill] (2.6,0) circle [radius=1.5pt];
\draw [fill] (5.4,0) circle [radius=1.5pt];
\draw [fill] (9.6,0) circle [radius=1.5pt];
\draw [fill] (12.2,0) circle [radius=1.5pt];

\draw [fill] (14,0) circle [radius=1.5pt];
\draw [fill] (14.8,0) circle [radius=1.5pt];
\draw [fill] (0.6,1) circle [radius=1.5pt];
\draw [fill] (1.7,1) circle [radius=1.5pt];
\draw [fill] (3.7,1) circle [radius=1.5pt];
\draw [fill] (11.1,1) circle [radius=1.5pt];
\draw [fill] (13.1,1) circle [radius=1.5pt];

\draw [fill] (14.4,1) circle [radius=1.5pt];
\draw
node at (7.5,-1) [] {\dots\dots\dots}
node at (7.5,1.9) [] {\dots\dots\dots}
node at (1.3,2.6) [] (e1) {$\epsilon$}
node at (1.3,1) [] (e2) {}
node at (13.75,2.6) [] (e3) {$1-\epsilon$}
node at (13.75,1) [] (e4) {}
node at (0.2,1.7) [] (x1) {}
node at (0,-1) []  {\small $H_{A_0}^{(m+1)}$}
node at (0.6,1.7) [] (x2) {}
node at (0.6,1.9) []  {\small $H_{A_0}^{(m)}$}
node at (1,1.7) [] (x3) {}
node at (1.1,-1) []  {\small $H_{A_1}^{(m+1)}$}

node at (2.15,0.7) [] {\small $>\delta$}
node at (1.7,1.7) [] (x6) {}
node at (1.8,1.9) []  {\small $H_{A_1}^{(m)}$}
node at (2.6,1.7) [] (x7) {}
node at (2.6,-1) []  {\small $H_{A_2}^{(m+1)}$}
node at (3.15,0.7) [] {\small $>\delta$}
node at (3.7,1.7) [] (x8) {}
node at (3.7,1.9) []  {\small $H_{A_2}^{(m)}$}
node at (4.55,0.7) [] {\small $>\delta$}
node at (5.4,1.7) [] (x9) {}
node at (5.4,-1) []  {\small $H_{A_3}^{(m+1)}$}
node at (9.6,1.7) [] (x10) {}
node at (9.6,-1) []  {\small $H_{A_{m-2}}^{(m+1)}$}
node at (10.35,0.7) [] {\small $>\delta$}
node at (11.1,1.7) [] (x11) {}
node at (11.1,1.9) []  {\small $H_{A_{m-2}}^{(m)}$}
node at (11.65,0.7) [] {\small $>\delta$}
node at (12.2,1.7) [] (x12) {}
node at (12.2,-1) []  {\small $H_{A_{m-1}}^{(m+1)}$}
node at (12.65,0.7) [] {\small $>\delta$}
node at (13.1,1.7) [] (x13) {}
node at (13.1,1.9) []  {\small $H_{A_{m-1}}^{(m)}$}

node at (14,1.7) [] (x16) {}
node at (14,-1) []  {\small $H_{A_m}^{(m+1)}$}
node at (14.4,1.7) [] (x17) {}
node at (14.4,1.9) []  {\small $H_{A_m}^{(m)}$}
node at (14.8,1.7) [] (x18) {}
node at (15.1,-1) []  {\small $H_{A_{m+1}}^{(m+1)}$}
node at (0.2,-0.7) [] (y1) {}
node at (0.6,-0.7) [] (y2) {}
node at (1,-0.7) [] (y3) {}

node at (1.6,0.4) [] (m1) {}
node at (2.7,0.4) [] (m2) {}
node at (2.5,0.4) [] (m3) {}
node at (3.8,0.4) [] (m4) {}
node at (3.6,0.4) [] (m5) {}
node at (5.5,0.4) [] (m6) {}

node at (1.7,-0.7) [] (y6) {}
node at (2.6,-0.7) [] (y7) {}
node at (3.7,-0.7) [] (y8) {}
node at (5.4,-0.7) [] (y9) {}
node at (9.6,-0.7) [] (y10) {}
node at (11.1,-0.7) [] (y11) {}
node at (12.2,-0.7) [] (y12) {}
node at (13.1,-0.7) [] (y13) {}

node at (9.5,0.4) [] (m11) {}
node at (11.2,0.4) [] (m12) {}
node at (11,0.4) [] (m13) {}
node at (12.3,0.4) [] (m14) {}
node at (12.1,0.4) [] (m15) {}
node at (13.2,0.4) [] (m16) {}

node at (14,-0.7) [] (y16) {}
node at (14.4,-0.7) [] (y17) {}
node at (14.8,-0.7) [] (y18) {};
\draw [densely dashed] (x1) -- (y1);
\draw [densely dashed] (x2) -- (y2);
\draw [densely dashed] (x3) -- (y3);

\draw [densely dashed] (x6) -- (y6);
\draw [densely dashed] (x7) -- (y7);
\draw [densely dashed] (x8) -- (y8);
\draw [densely dashed] (x9) -- (y9);
\draw [densely dashed] (x10) -- (y10);
\draw [densely dashed] (x11) -- (y11);
\draw [densely dashed] (x12) -- (y12);
\draw [densely dashed] (x13) -- (y13);

\draw [densely dashed] (x16) -- (y16);
\draw [densely dashed] (x17) -- (y17);
\draw [densely dashed] (x18) -- (y18);
\draw [<->] (m1) -- (m2);
\draw [<->] (m3) -- (m4);
\draw [<->] (m5) -- (m6);

\draw [<->] (m11) -- (m12);
\draw [<->] (m13) -- (m14);
\draw [<->] (m15) -- (m16);
\draw [->] (e1) -- (e2);
\draw [->] (e3) -- (e4);
\end{tikzpicture}
\caption{Illustration of the interlacing property in \eqref{eq:jjw1} used in the proofs of Theorem~\ref{thm:m1} and Theorem~\ref{thm:m2}.}
\label{fig:bvd}
\end{figure}
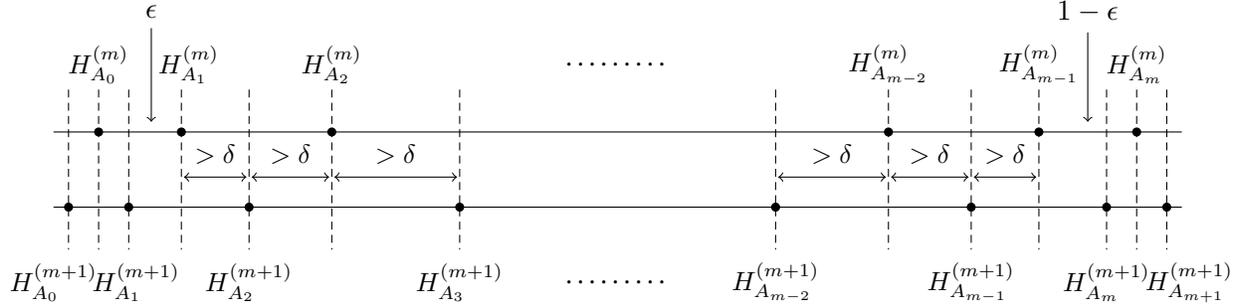

\begin{figure}[h]
\centering
\includegraphics[width=0.8\textwidth]{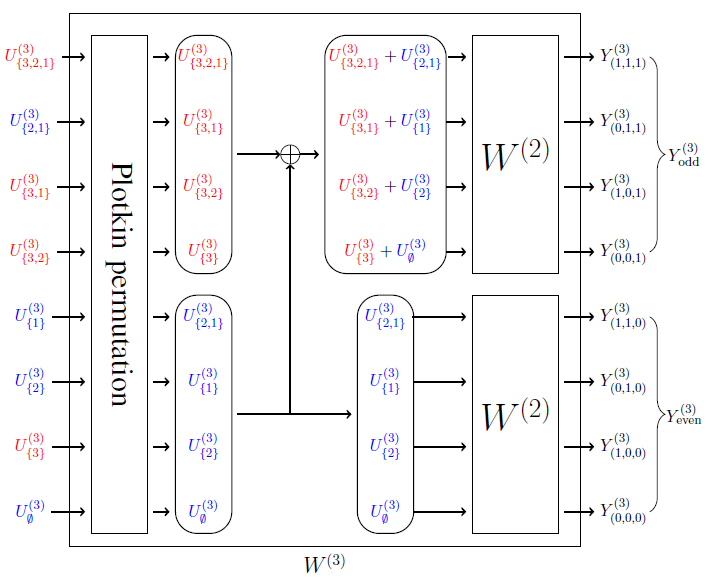}
\caption{The conditional distribution of $Y_{\odd}^{(3)}$ given 
$(U_{A}^{(3)}+U_{A\cup\{3\}}^{(3)}, A\subseteq[2])$ is exactly $W^{(2)}$, and so is the conditional distribution of $Y_{\even}^{(3)}$ given $(U_{A}^{(3)}, A\subseteq[2])$.}
\label{fig:m=2}
\end{figure}

\begin{figure}[h]
\centering
\begin{subfigure}{.45\linewidth}
\begin{tikzpicture}
\draw
 node at (0,10.5) [] (u1)  {\textcolor{red}{$U_{\{3,2,1\}}^{(3)}$}}
 node at (0,9) [] (u2)  {$\textcolor{blue}{U_{\{2,1\}}^{(3)}}$}
 node at (0,7.5) [] (u3) {\textcolor{red}{$U_{\{3,1\}}^{(3)}$}}
 node at (0,6) []  (u4) {$\textcolor{red}{U_{\{3,2\}}^{(3)}}$}
 node at (0,4.5) []  (u5) {$\textcolor{blue}{U_{\{1\}}^{(3)}}$}
 node at (0,3) []  (u6) {$\textcolor{blue}{U_{\{2\}}^{(3)}}$}
 node at (0,1.5) [] (u7)  {$\textcolor{red}{U_{\{3\}}^{(3)}}$}
 node at (0,0) [] (u8)  {$\textcolor{blue}{U_{\emptyset}^{(3)}}$}
 
 node at (2,10.5) [XOR,scale=1.2] (x1) {}
 node at (2, 6.75) [XOR,scale=1.2] (x2) {}
 node at (2, 1.5) [XOR,scale=1.2] (x3) {}
 node at (3.5,10.5) [block] (v1)  {$W_{\{2,1\}}^{(2)}$}
 node at (3.5,9) [block] (v2)  {$W_{\{2,1\}}^{(2)}$}
 node at (5.5,10.5) [] (y1)  {$Y_{2,\odd}^{(3)}$}
 node at (5.5,9) [] (y2)  {$Y_{2,\even}^{(3)}$}
 
 node at (3.5,7.5) []   {$W_{\{1\}}^{(2)}$}
 node at (3.5,6) []   {$W_{\{2\}}^{(2)}$}
 node at (5.5,6.75) [] (y3) {$Y_{1,\odd}^{(3)}$}
 
 node at (3.5,4.5) []   {$W_{\{1\}}^{(2)}$}
 node at (3.5,3) []   {$W_{\{2\}}^{(2)}$}
 node at (5.5,3.75) [] (y5) {$Y_{1,\even}^{(3)}$}
 
 node at (3.5,1.5) [block] (v7)  {$W_{\emptyset}^{(2)}$}
 node at (3.5,0) [block] (v8)  {$W_{\emptyset}^{(2)}$}
 node at (5.5,1.5) [] (y7)  {$Y_{0,\odd}^{(3)}$}
 node at (5.5,0) [] (y8)  {$Y_{0,\even}^{(3)}$};

 \draw [thick, rounded corners=5mm] (-0.6, 5.5) rectangle (0.6,8);
 \draw [thick, rounded corners=5mm] (-0.6, 2.5) rectangle (0.6,5);
 \draw [thick] (2.95, 5.5) rectangle (4.05,8);
 \draw [thick] (2.95, 2.5) rectangle (4.05,5);

 \draw[very thick,->](u1) -- node {}(x1);
 \draw[very thick,->](u2) -| node {}(x1);
 \draw[very thick,->](x1) -- node {}(v1);
 \draw[very thick,->](u2) -- node {}(v2);
 \draw[very thick,->](v1) -- node {}(y1);
 \draw[very thick,->](v2) -- node {}(y2);
 \draw[very thick,->](4.05, 6.75) -- node {}(y3);
 \draw[very thick,->](4.05, 3.75) -- node {}(y5);
 \draw[very thick,->](v7) -- node {}(y7);
 \draw[very thick,->](v8) -- node {}(y8);

 \draw[very thick,->](u7) -- node {}(x3);
 \draw[very thick,->](u8) -| node {}(x3);
 \draw[very thick,->](x3) -- node {}(v7);
 \draw[very thick,->](u8) -- node {}(v8);

 \draw[very thick,->](0.6, 6.75) -- node {}(x2);
 \draw[very thick,->](0.6, 3.75) -| node {}(x2);
 \draw[very thick,->](x2) -- node {}(2.95, 6.75);
 \draw[very thick,->](0.6, 3.75) -- node {}(2.95, 3.75);
 
\end{tikzpicture}
\end{subfigure}
\hspace*{0.2in}
\begin{subfigure}{.45\linewidth}
\begin{tikzpicture}
\draw
 node at (0,10.5) [] (u1)  {$W_{\{2,1\}}^{(2)}$}
 node at (0,9) [] (u2)  {$W_{\{2,1\}}^{(2)}$}
 node at (1.2,10.5) [] (v1)  {}
 node at (1.2,9) [] (v2)  {}
 node at (3.7,10.5) [] (x1)  {}
 node at (3.7,9) [] (x2)  {}
 node at (3.7,10.5) [] (x1)  {}
 node at (3.7,9) [] (x2)  {}
 node at (5.2,10.5) [] (y1)  {$W_{\{3,2,1\}}^{(3)}$}
 node at (5.2,9) [] (y2)  {$W_{\{2,1\}}^{(3)}$}
 node at (2.45,9.75) [text width=2cm,align=center] {Polar\\Transform};
 \draw[thick] (1.2,8.5) rectangle (3.7,11);
 \draw[very thick,->](u1) -- node {}(v1);
 \draw[very thick,->](u2) -- node {}(v2);
 \draw[very thick,->](x1) -- node[above] {``$-$"}(y1);
 \draw[very thick,->](x2) -- node [above] {``$+$"} (y2);

 \draw
 node at (0,7.5) []   {$W_{\{1\}}^{(2)}$}
 node at (0,6) []   {$W_{\{2\}}^{(2)}$}
 node at (0,4.5) []   {$W_{\{1\}}^{(2)}$}
 node at (0,3) []   {$W_{\{2\}}^{(2)}$}

 node at (5.2,7.5) []   {$W_{\{3,1\}}^{(3)}$}
 node at (5.2,6) []   {$W_{\{3,2\}}^{(3)}$}
 node at (5.2,4.5) []   {$W_{\{1\}}^{(3)}$}
 node at (5.2,3) []   {$W_{\{2\}}^{(3)}$}
 node at (2.45,5.25) [text width=2cm,align=center] {Fast\\Polar\\Transform};
 
 \draw [thick, rounded corners=5mm] (-0.6, 5.5) rectangle (0.6,8);
 \draw [thick, rounded corners=5mm] (-0.6, 2.5) rectangle (0.6,5);
 \draw [thick, rounded corners=5mm] (4.6, 5.5) rectangle (5.8,8);
 \draw [thick, rounded corners=5mm] (4.6, 2.5) rectangle (5.8,5);
 
 \draw[thick] (1.2, 3) rectangle (3.7, 7.5);
 \draw[very thick,->](0.6, 6.75) -- node {}(1.2, 6.75);
 \draw[very thick,->](0.6, 3.75) -- node {}(1.2, 3.75);
 \draw[very thick,->](3.7, 6.75) -- node[above] {``$-$"}(4.6, 6.75);
 \draw[very thick,->](3.7, 3.75) -- node [above] {``$+$"} (4.6, 3.75);

 \draw
 node at (0,1.5) [] (u1)  {$W_\emptyset^{(2)}$}
 node at (0,0) [] (u2)  {$W_\emptyset^{(2)}$}
 node at (1.2,1.5) [] (v1)  {}
 node at (1.2,0) [] (v2)  {}
 node at (3.7,1.5) [] (x1)  {}
 node at (3.7,0) [] (x2)  {}
 node at (3.7,1.5) [] (x1)  {}
 node at (3.7,0) [] (x2)  {}
 node at (5.2,1.5) [] (y1)  {$W_{\{3\}}^{(3)}$}
 node at (5.2,0) [] (y2)  {$W_\emptyset^{(3)}$}
 node at (2.45,0.75) [text width=2cm,align=center] {Polar\\Transform};
 \draw[thick] (1.2,-0.5) rectangle (3.7,2);
 \draw[very thick,->](u1) -- node {}(v1);
 \draw[very thick,->](u2) -- node {}(v2);
 \draw[very thick,->](x1) -- node[above] {``$-$"}(y1);
 \draw[very thick,->](x2) -- node [above] {``$+$"} (y2);
\end{tikzpicture}
\end{subfigure}
\caption{The channel outputs are given in \eqref{eq:toe}.
$W_{\{3,2,1\}}^{(3)}$ is the channel mapping from $U_{\{3,2,1\}}^{(3)}$ to $(Y_{2,\odd}^{(3)},Y_{2,\even}^{(3)})$, so it is the ``$-$" polar transform of $W_{\{2,1\}}^{(2)}$. $W_{\{2,1\}}^{(3)}$ is the channel mapping from $U_{\{2,1\}}^{(3)}$ to $(U_{\{3,2,1\}}^{(3)},Y_{2,\odd}^{(3)},Y_{2,\even}^{(3)})$, so it is the ``$+$" polar transform of $W_{\{2,1\}}^{(2)}$.
Similarly, the outputs of $W_{\{3\}}^{(3)}$ are bijections of $(Y_{0,\odd}^{(3)},Y_{0,\even}^{(3)})$, and the outputs of $W_\emptyset^{(3)}$ are bijections of $(U_{\{3\}}^{(3)},Y_{0,\odd}^{(3)},Y_{0,\even}^{(3)})$, so $W_{\{3\}}^{(3)}$ and $W_\emptyset^{(3)}$ are polar transforms of $W_\emptyset^{(2)}$.
See Fig.~\ref{fig:bpt} for explanations of fast polar transform and the bit-channels
$W_{\{3,1\}}^{(3)}, W_{\{3,2\}}^{(3)}, W_{\{1\}}^{(3)}, W_{\{2\}}^{(3)}$.}
\label{fig:fceng}
\end{figure}
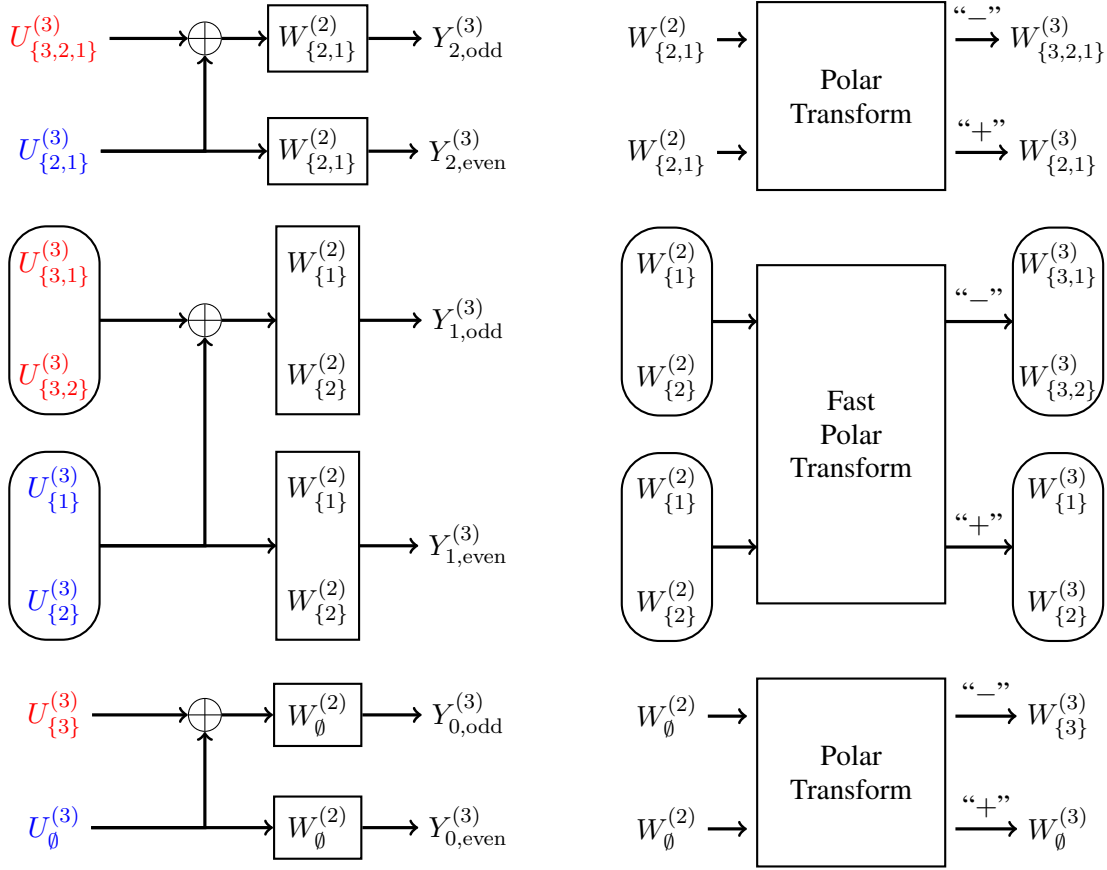

\begin{figure}
    \centering
\begin{tikzpicture}
\draw [thick] (0,0) rectangle (0.45, 0.6);
\draw [thick] (0.45,0) rectangle (2.7, 0.6);
\draw [thick] (2.7,0) rectangle (7.2, 0.6);
\draw [thick] (7.2,0) rectangle (11.7, 0.6);
\draw [thick] (11.7,0) rectangle (13.95, 0.6);
\draw [thick] (13.95,0) rectangle (14.4, 0.6);

\draw [decorate,decoration={brace,amplitude=10pt},xshift=0pt,yshift=0pt] (2.7, 0) -- (0.45, 0) node [black,midway,yshift=-0.6cm]   {Layer $4$};
\draw [decorate,decoration={brace,amplitude=10pt},xshift=0pt,yshift=0pt] (7.2, 0) -- (2.7, 0) node [black,midway,yshift=-0.6cm]   {Layer $3$};
\draw [decorate,decoration={brace,amplitude=10pt},xshift=0pt,yshift=0pt] (11.7, 0) -- (7.2, 0) node [black,midway,yshift=-0.6cm]   {Layer $2$};
\draw [decorate,decoration={brace,amplitude=10pt},xshift=0pt,yshift=0pt] (13.95, 0) -- (11.7, 0) node [black,midway,yshift=-0.6cm]   {Layer $1$};

\draw
node at (0.225, -1.1) [] (l5) {Layer $5$}
node at (14.175, -1.1) [] (l0) {Layer $0$};

\draw [very thick, ->] (l5) -- (0.225, 0);
\draw [very thick, ->] (l0) -- (14.175, 0);

\draw [densely dashed] (0.9, 0) -- (0.9, 0.6);
\draw [densely dashed] (4.5, 0) -- (4.5, 0.6);
\draw [densely dashed] (9.9, 0) -- (9.9, 0.6);
\draw [densely dashed] (13.5, 0) -- (13.5, 0.6);

\draw [thick] (0.225,2) rectangle (0.675, 2.6);
\draw [thick] (1.8,2) rectangle (3.6, 2.6);
\draw [thick] (5.85, 2) rectangle (8.55, 2.6);
\draw [thick] (10.8,2) rectangle (12.6, 2.6);
\draw [thick] (13.725,2) rectangle (14.175, 2.6);

\node at (0.45, 2.9) {Layer $4$};
\node at (2.7, 2.9) {Layer $3$};
\node at (7.2, 2.9) {Layer $2$};
\node at (11.7, 2.9) {Layer $1$};
\node at (13.95, 2.9) {Layer $0$};

\node at (7.2, 2.3) {$W_A^{(4)},|A|=2$};
\node at (5.8, 0.3) {\small $W_{A\cup\{5\}}^{(5)},|A|=2$};
\node at (8.55, 0.3) {$W_A^{(5)},|A|=2$};

\draw [very thick, ->] (0.45, 2) -- node [left] {``$-$"} (0.225, 0.6);
\draw [very thick, ->] (0.45, 2) -- node [right] {``$+$"} (0.675, 0.6);

\draw [very thick, ->] (2.7, 2) -- node [left] {``$-$"} (1.8, 0.6);
\draw [very thick, ->] (2.7, 2) -- node [right] {``$+$"} (3.6, 0.6);

\draw [very thick, ->] (7.2, 2) -- node [left] {``$-$"} (5.85, 0.6);
\draw [very thick, ->] (7.2, 2) -- node [right] {``$+$"} (8.55, 0.6);

\draw [very thick, ->] (11.7, 2) -- node [left] {``$-$"} (10.8, 0.6);
\draw [very thick, ->] (11.7, 2) -- node [right] {``$+$"} (12.6, 0.6);

\draw [very thick, ->] (13.95, 2) -- node [left] {``$-$"} (13.725, 0.6);
\draw [very thick, ->] (13.95, 2) -- node [right] {``$+$"} (14.175, 0.6);

\node at (15.2, 2.3) {$m=4$};
\node at (15.2, 0.3) {$m=5$};
\end{tikzpicture}
\caption{The evolution from level $m=4$ to level $m=5$. The $i$-th layer $\{W_A^{(5)},A\subseteq[5],|A|=i\}$ in the $5$-th level consists of the ``$+$" fast polar transform of the $i$-th layer $\{W_A^{(4)},A\subseteq[4],|A|=i\}$ in the $4$-th level and the ``$-$" fast polar transform of the $(i-1)$-th layer $\{W_A^{(4)},A\subseteq[4],|A|=i-1\}$ in the $4$-th level.}
\label{fig:rs}

\end{figure}
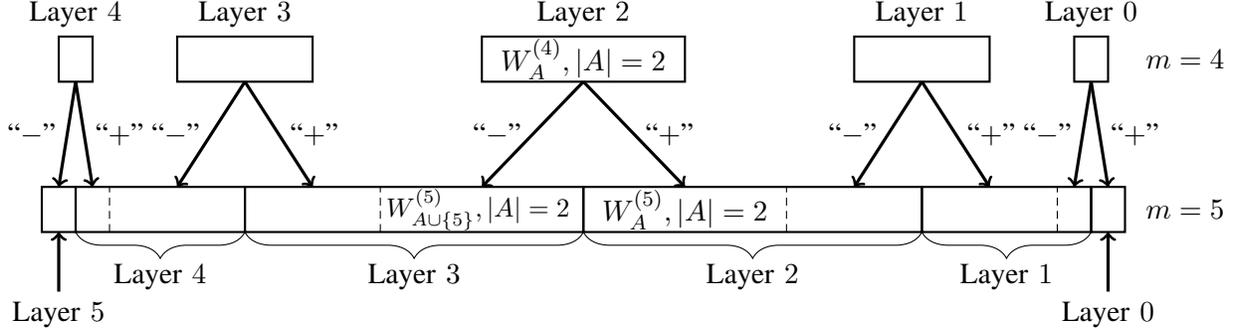

\begin{figure}
\centering
\begin{subfigure}{.45\linewidth}
\begin{tikzpicture}
\node at (0,4) {$\begin{array}{c}
   U_{B_1\cup\{m+1\}}^{(m+1)} \vspace*{0.1in}  \\
   U_{B_2\cup\{m+1\}}^{(m+1)}  \\
    \vdots  \\
    \vdots  \\
   U_{B_j\cup\{m+1\}}^{(m+1)}
\end{array}$};

\node at (0,0) {$\begin{array}{c}
   U_{B_1}^{(m+1)}  \vspace*{0.1in} \\
   U_{B_2}^{(m+1)}  \\
    \vdots  \\
    \vdots  \\
   U_{B_j}^{(m+1)}
\end{array}$};

\node at (3.5,0) {$\begin{array}{c}
   W_{B_1}^{(m)}  \vspace*{0.1in} \\
   W_{B_2}^{(m)}  \\
    \vdots  \\
    \vdots  \\
   W_{B_j}^{(m)}
\end{array}$};

\node at (3.5,4) {$\begin{array}{c}
   W_{B_1}^{(m)}  \vspace*{0.1in} \\
   W_{B_2}^{(m)}  \\
    \vdots  \\
    \vdots  \\
   W_{B_j}^{(m)}
\end{array}$};

 \node at (5.5, 4) [] (y1)  {$Y_{t,\odd}^{(m+1)}$};
 \node at (5.5, 0) [] (y2)  {$Y_{t,\even}^{(m+1)}$};

 \node at (2, 4) [XOR,scale=1.2] (x1) {};
 \draw [thick, rounded corners=5mm] (-0.8, -1.8) rectangle (0.8, 1.8);
 \draw [thick, rounded corners=5mm] (-1, 2.2) rectangle (1, 5.8);
 \draw[thick] (3,-1.8) rectangle (4,1.8);
 \draw[thick] (3,2.2) rectangle (4,5.8);

 \draw[very thick,->](1,4) -- node {}(x1);
 \draw[very thick,->](0.8, 0) -| node {}(x1);
 \draw[very thick,->](x1) -- node {}(3, 4);
 \draw[very thick,->](0.8, 0) -- node {}(3, 0);
 \draw[very thick,->](4, 4) -- node {}(y1);
 \draw[very thick,->](4, 0) -- node {}(y2);
\end{tikzpicture}
\end{subfigure}
\hspace*{0.2in}
\begin{subfigure}{.45\linewidth}
\begin{tikzpicture}
\node at (5.2, 4) {$\begin{array}{c}
   W_{B_1\cup\{m+1\}}^{(m+1)} \vspace*{0.1in}  \\
   W_{B_2\cup\{m+1\}}^{(m+1)}  \\
    \vdots  \\
    \vdots  \\
   W_{B_j\cup\{m+1\}}^{(m+1)}
\end{array}$};

\node at (5, 0) {$\begin{array}{c}
   W_{B_1}^{(m+1)}  \vspace*{0.1in} \\
   W_{B_2}^{(m+1)}  \\
    \vdots  \\
    \vdots  \\
   W_{B_j}^{(m+1)}
\end{array}$};

\node at (0,0) {$\begin{array}{c}
   W_{B_1}^{(m)}  \vspace*{0.1in} \\
   W_{B_2}^{(m)}  \\
    \vdots  \\
    \vdots  \\
   W_{B_j}^{(m)}
\end{array}$};

\node at (0,4) {$\begin{array}{c}
   W_{B_1}^{(m)}  \vspace*{0.1in} \\
   W_{B_2}^{(m)}  \\
    \vdots  \\
    \vdots  \\
   W_{B_j}^{(m)}
\end{array}$};

\node at (2.4, 2) [text width=2cm,align=center] {Fast\\Polar\\Transform};
 \draw[thick] (1.4,-1) rectangle (3.4, 5);
 \draw [thick, rounded corners=5mm] (-0.6, -1.8) rectangle (0.6, 1.8);
 \draw [thick, rounded corners=5mm] (-0.6, 2.2) rectangle (0.6, 5.8);
 \draw [thick, rounded corners=5mm] (4.2, -1.8) rectangle (5.8, 1.8);
 \draw [thick, rounded corners=5mm] (4.2, 2.2) rectangle (6.3, 5.8);

 \draw[very thick,->](0.6, 4) -- node {}(1.4, 4);
 \draw[very thick,->](0.6, 0) -- node {}(1.4, 0);
 \draw[very thick,->](3.4, 4) -- node [above] {``$-$"}(4.2, 4);
 \draw[very thick,->](3.4, 0) -- node [above] {``$+$"}(4.2, 0);

\end{tikzpicture}
\end{subfigure}
\caption{Fast Polar Transform. $B_1<B_2<\dots<B_j$ are all the subsets of $[m]$ with cardinality $t$. $Y_{t,\odd}^{(m+1)}$ and $Y_{t,\even}^{(m+1)}$ are given in \eqref{eq:toe}.}
\label{fig:bpt}
\end{figure}
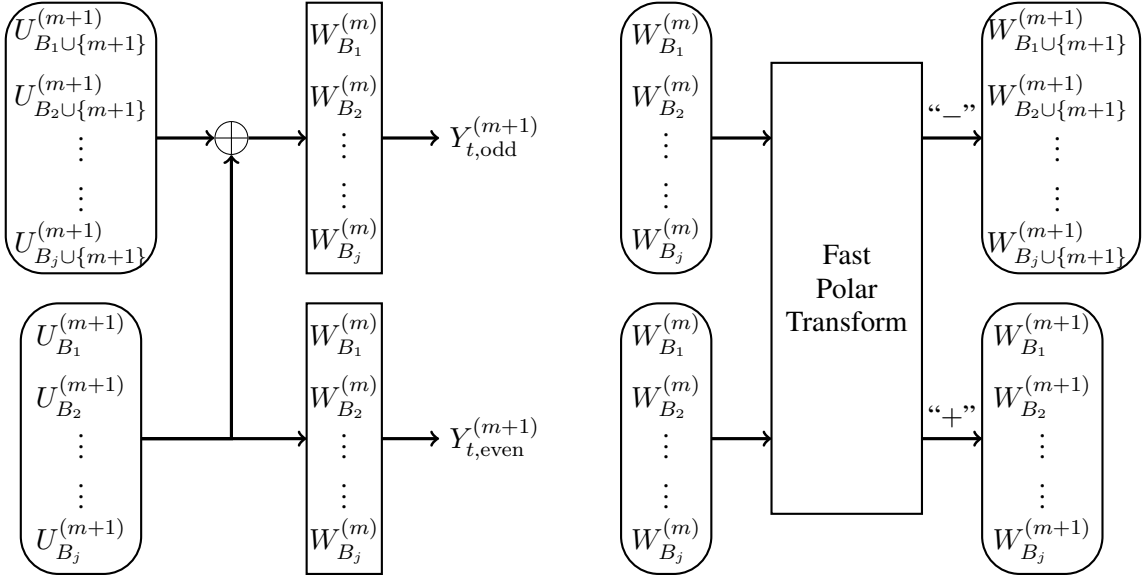

\begin{figure}[h]
\centering
\begin{subfigure}{.45\linewidth}
\begin{tikzpicture}
\draw
 node at (0, 4.5) [] (u1)  {$U_{B_i\cup\{m+1\}}^{(m+1)}$}
 node at (0,3) [] (u2)  {$U_{B_i}^{(m+1)}$}
 node at (2,4.5) [XOR,scale=1.2] (x1) {}
 node at (3.5,4.5) [block] (v1)  {$W_{B_i}^{(m)}$}
 node at (3.5,3) [block] (v2)  {$W_{B_i}^{(m)}$}
 node at (5.5,4.5) [] (y1)  {$Y_{B_i,\odd}^{(m+1)}$}
 node at (5.5,3) [] (y2)  {$Y_{B_i,\even}^{(m+1)}$};
 \draw[very thick,->](u1) -- node {}(x1);
 \draw[very thick,->](u2) -| node {}(x1);
 \draw[very thick,->](x1) -- node {}(v1);
 \draw[very thick,->](u2) -- node {}(v2);
 \draw[very thick,->](v1) -- node {}(y1);
 \draw[very thick,->](v2) -- node {}(y2);
\end{tikzpicture}
\end{subfigure}
\hspace*{0.2in}
\begin{subfigure}{.45\linewidth}
\begin{tikzpicture}
\draw
 node at (0,1.5) [] (u1)  {$W_{B_i}^{(m)}$}
 node at (0,0) [] (u2)  {$W_{B_i}^{(m)}$}
 node at (1.2,1.5) [] (v1)  {}
 node at (1.2,0) [] (v2)  {}
 node at (3.7,1.5) [] (x1)  {}
 node at (3.7,0) [] (x2)  {}
 node at (5.4,1.5) [] (y1)  {$\widetilde{W}_{B_i\cup\{m+1\}}^{(m+1)}$}
 node at (5.2,0) [] (y2)  {$\widetilde{W}_{B_i}^{(m+1)}$}
 node at (2.45,0.75) [text width=2cm,align=center] {Polar\\Transform};
 \draw[thick] (1.2,-0.5) rectangle (3.7,2);
 \draw[very thick,->](u1) -- node {}(v1);
 \draw[very thick,->](u2) -- node {}(v2);
 \draw[very thick,->](x1) -- node[above] {``$-$"}(y1);
 \draw[very thick,->](x2) -- node [above] {``$+$"} (y2);

\end{tikzpicture}
\end{subfigure}
\caption{$Y_{B_i,\odd}^{(m+1)}$ and $Y_{B_i,\even}^{(m+1)}$ are defined in \eqref{eq:lby}.
The outputs of $\widetilde{W}_{B_i\cup\{m+1\}}^{(m+1)}$ are bijections of $(Y_{B_i,\odd}^{(m+1)},Y_{B_i,\even}^{(m+1)})$, and the outputs of $\widetilde{W}_{B_i}^{(m+1)}$ are bijections of $(U_{B_i\cup\{m+1\}}^{(m+1)},Y_{B_i,\odd}^{(m+1)},Y_{B_i,\even}^{(m+1)})$, so $\widetilde{W}_{B_i\cup\{m+1\}}^{(m+1)}$ and $\widetilde{W}_{B_i}^{(m+1)}$ are ``$-$" and ``$+$" polar transforms of $W_{B_i}^{(m)}$, respectively.}
\label{fig:1and2}
\end{figure}
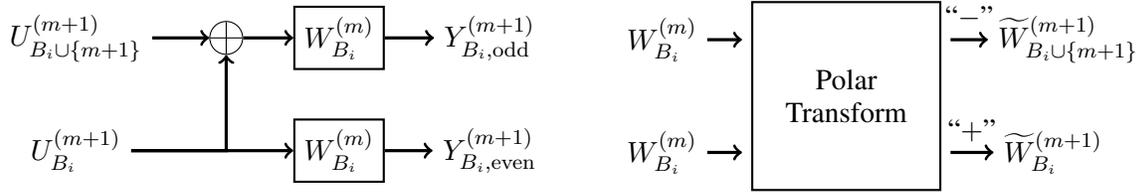

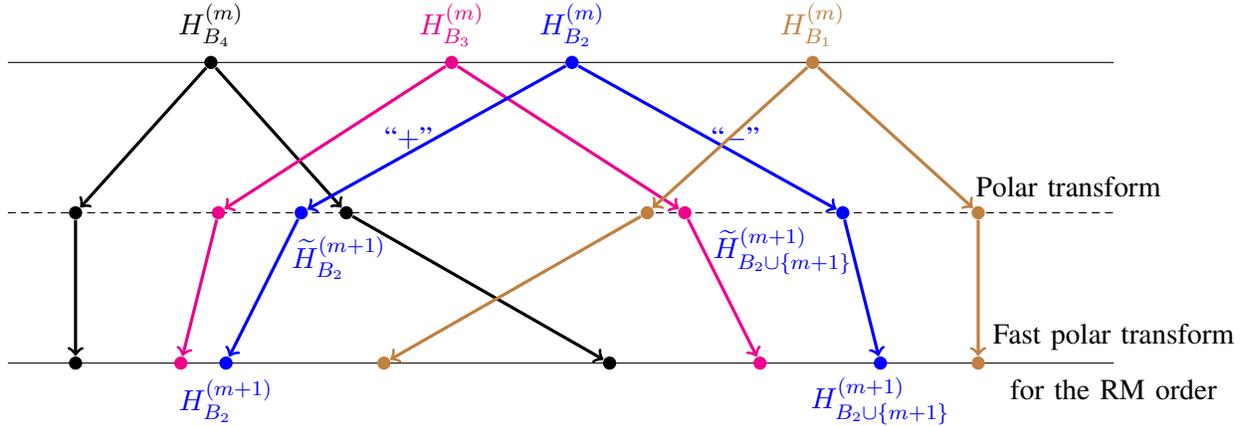
\begin{figure}
\centering
\begin{tikzpicture}[every node/.style={circle,inner sep=0pt, minimum size= 5pt}]
\draw (0.3,0) -- (15,0);
\draw [densely dashed] (0.3, 2) -- (15, 2);
\draw (0.3, 4) -- (15, 4);

\node at (14.4, 2.35) {Polar transform};
\node at (15, 0.35) {Fast polar transform};
\node at (15, -0.35) {for the RM order};

\node at (3, 4.5) {$H_{B_4}^{(m)}$};
\node at (6.2, 4.5) [color=magenta] {$H_{B_3}^{(m)}$};
\node at (7.8, 4.5) [color=blue] {$H_{B_2}^{(m)}$};
\node at (11, 4.5) [color=brown] {$H_{B_1}^{(m)}$};

\node at (3.2, -0.5) [color=blue] {$H_{B_2}^{(m+1)}$};
\node at (11.9, -0.5) [color=blue] {$H_{B_2\cup\{m+1\}}^{(m+1)}$};

\node at (4.7, 1.4) [color=blue] {$\widetilde{H}_{B_2}^{(m+1)}$};
\node at (10.6, 1.5) [color=blue] {$\widetilde{H}_{B_2\cup\{m+1\}}^{(m+1)}$};

\node at (3,4) [fill]  (bk0) {};
\node at (6.2,4) [fill= magenta]  (m0) {};
\node at (7.8,4) [fill= blue]  (bl0) {};
\node at (11,4) [fill= brown]  (br0) {};

\node at (1.2,2) [fill]  (bk1) {};
\node at (4.8,2) [fill] (bk2) {};

\node at (3.1,2) [fill, magenta] (m1) {};
\node at (9.3,2) [fill, magenta] (m2) {};

\node at (4.2,2) [fill, blue] (bl1) {};
\node at (11.4,2) [fill, blue] (bl2) {};

\node at (8.8,2) [fill, brown] (br1) {};
\node at (13.2,2) [fill, brown] (br2) {};

\node at (1.2, 0) [fill] (bk3) {};
\node at (8.3, 0) [fill] (bk4) {};

\node at (2.6,0) [fill, magenta] (m3) {};
\node at (10.3,0) [fill, magenta] (m4) {};

\node at (3.2,0) [fill, blue] (bl3) {};
\node at (11.9, 0) [fill, blue] (bl4) {};

\node at (5.3,0) [fill, brown] (br3) {};
\node at (13.2,0) [fill, brown] (br4) {};

\draw[very thick,->](bk0) -- node {}(bk1);
\draw[very thick,->](bk0) -- node {}(bk2);
\draw[very thick,->](bk1) -- node {}(bk3);
\draw[very thick,->](bk2) -- node {}(bk4);

\draw[very thick,->, color=magenta](m0) -- node {}(m1);
\draw[very thick,->, color=magenta](m0) -- node {}(m2);
\draw[very thick,->, color=magenta](m1) -- node {}(m3);
\draw[very thick,->, color=magenta](m2) -- node {}(m4);

\draw[very thick,->, color=blue](bl0) -- node [left] {``$+$"}(bl1);
\draw[very thick,->, color=blue](bl0) -- node [right] {``$-$"} (bl2);
\draw[very thick,->, color=blue](bl1) -- node {}(bl3);
\draw[very thick,->, color=blue](bl2) -- node {}(bl4);

\draw[very thick,->, color=brown](br0) -- node {}(br1);
\draw[very thick,->, color=brown](br0) -- node {}(br2);
\draw[very thick,->, color=brown](br1) -- node {}(br3);
\draw[very thick,->, color=brown](br2) -- node {}(br4);
\end{tikzpicture}
\caption{The fast polar transform with block size $4$.
The dots on the second line are the results of the standard polar transform, and the dots on the third line are the results of fast polar transform. In the fast polar transform, the worse (``$-$") bit-channel in the standard polar transform gets even worse, and the better (``$+$") bit-channel in the standard polar transform gets even better. Therefore, the gap between $H_{B_i\cup\{m+1\}}^{(m+1)}$ and $H_{B_i}^{(m+1)}$ is always larger than the gap between $\widetilde{H}_{B_i\cup\{m+1\}}^{(m+1)}$ and $\widetilde{H}_{B_i}^{(m+1)}$. Intuitively, this explains why RM codes polarize and do so even faster than polar codes.
}
\label{fig:gp}
\end{figure}

In order to prove these two theorems, we only need to show two results. We establish  first an {\it interlacing property:}
\begin{equation} \label{eq:jjw1}
\text{{\it Interlacing property:}} \quad
 H_{A_i}^{(m+1)} \le H_{A_i}^{(m)} \le H_{A_{i+1}}^{(m+1)}  \quad \forall i\in\{0,1,\dots,m\}.
\end{equation}
Second, for any $\epsilon>0$, there is $\delta(\epsilon)>0$ such that for any increasing chain of sets and any $i\in\{0,1,\dots,m\}$,
$$
H_{A_i}^{(m)} \in (\epsilon,1-\epsilon)
$$
implies that
\begin{equation} \label{eq:jjw2}
 H_{A_i}^{(m)} -  H_{A_i}^{(m+1)} > \delta(\epsilon)
\text{~and~} 
 H_{A_{i+1}}^{(m+1)} - H_{A_i}^{(m)} > \delta(\epsilon).
\end{equation}

It is clear that Theorem~\ref{thm:m1} follows immediately from \eqref{eq:jjw1}; see Fig.~\ref{fig:bvd} for an illustration. Now we prove Theorem~\ref{thm:m2} using \eqref{eq:jjw2} and Theorem~\ref{thm:m1}.
By \eqref{eq:jjw2} we know that as long as $H_{A_i}^{(m)}>\epsilon$ and $H_{A_{i+1}}^{(m)}<1-\epsilon$, we have $H_{A_{i+1}}^{(m)} - H_{A_i}^{(m)}>2\delta$; see Fig.~\ref{fig:bvd} for an illustration.
Let $j$ be the smallest index such that $H_{A_j}^{(m)}>\epsilon$, and let $j'$ be the largest index such that $H_{A_{j'}}^{(m)}<1-\epsilon$. Then
$$
\left|\left\{ i\in\{0,1,\dots,m\}:\epsilon < H_{A_i}^{(m)} < 1-\epsilon \right\} \right| = j'-j+1. 
$$
Since $H_{A_i}^{(m)}$ increases with $i$, we have 
$$
H_{A_{j'}}^{(m)}- H_{A_j}^{(m)}=\sum_{i=j}^{j'-1} (H_{A_{i+1}}^{(m)} - H_{A_i}^{(m)}) >2(j'-j)\delta.
$$
Since $H_{A_{j'}}^{(m)}- H_{A_j}^{(m)}$ is upper bounded by $1$, we have $j'-j<\frac{1}{2\delta}$. Therefore,
$$
\left|\left\{ i\in\{0,1,\dots,m\}:\epsilon < H_{A_i}^{(m)} < 1-\epsilon \right\} \right| < \frac{1}{2\delta}+1. 
$$
Thus we have proved Theorem~\ref{thm:m2} with the choice of $D(\epsilon)= \frac{1}{2\delta(\epsilon)}+1$.

Now we are left to explain how to prove \eqref{eq:jjw1}--\eqref{eq:jjw2}.
The proof is divided into two steps. First, we prove \eqref{eq:jjw1}--\eqref{eq:jjw2} for the special case of $A_{i+1}=A_i\cup\{m+1\}$. Then we show that by the symmetry of RM codes, $H_{A_i\cup\{j\}}^{(m+1)} \ge H_{A_i\cup\{m+1\}}^{(m+1)}$ for any $j\in[m]\setminus A_i$; see Lemma~\ref{lm:is}.
Below we focus on the explanation of the first part.

To prove \eqref{eq:jjw1}--\eqref{eq:jjw2} for the special case of $A_{i+1}=A_i\cup\{m+1\}$,
we use the recursive structure of RM code, and connect it back to that of polar codes. However, \eqref{eq:jjw1} is {\it not} a polar code triplet of the kind $W^- \le W \le W^+$ since we are working with the RM code ordering. The good news is that \eqref{eq:jjw1} gives in fact a larger spread than the one occurring for  triplets of polar codes. The rest of this section is dedicated to explaining the precise meaning of previous phrase. 
With this connection in mind, \eqref{eq:jjw1}--\eqref{eq:jjw2} will be derived from \eqref{eq:kbd} and Lemma~\ref{lm:sas}.

We now derive \eqref{eq:jjw1}--\eqref{eq:jjw2}.
For a given BMS channel $W$, we denote the channel mapping from $U^{(m)}$ to $Y^{(m,W)}$ as $W^{(m)}$.
Let us divide $Y^{(m+1)}:=(Y_{\mathbi{z}}^{(m+1)}:\mathbi{z}=(z_1,z_2,\dots,z_{m+1})\in\{0,1\}^{m+1})$ into two subsets:
\begin{equation} \label{eq:oddeven}
\begin{aligned}
Y_{\odd}^{(m+1)}:=(Y_{\mathbi{z}}^{(m+1)}:\mathbi{z}=(z_1,z_2,\dots,z_{m+1})\in\{0,1\}^{m+1},
z_{m+1} = 1) \\
Y_{\even}^{(m+1)}:= (Y_{\mathbi{z}}^{(m+1)}:\mathbi{z}=(z_1,z_2,\dots,z_{m+1})\in\{0,1\}^{m+1},
z_{m+1}=0).
\end{aligned}
\end{equation}
The main observation is that the conditional distribution of $Y_{\odd}^{(m+1)}$ given 
$(U_{A}^{(m+1)}+U_{A\cup\{m+1\}}^{(m+1)}, A\subseteq[m])$ is exactly $W^{(m)}$, and so is the conditional distribution of $Y_{\even}^{(m+1)}$ given $(U_{A}^{(m+1)}, A\subseteq[m])$.
The special case of $m=2$ is illustrated in Fig.~\ref{fig:m=2}. 
To see the connection to polar codes, we further analyze the relation between the bit-channels $\{W_A^{(3)}:A\subseteq[3]\}$ and $\{W_A^{(2)}:A\subseteq[2]\}$ in Fig.~\ref{fig:fceng}, where the outputs of the bit-channels in the figure are given by
\begin{equation} \label{eq:toe}
\begin{aligned}
    Y_{t,\odd}^{(m+1)} &:= (\{U_{B\cup\{m+1\}}^{(m+1)}+U_B^{(m+1)}:B\subseteq[m],|B|>t\},Y_{\odd}^{(m+1)}), \\
    Y_{t,\even}^{(m+1)} &:= (\{U_B^{(m+1)}:B\subseteq[m],|B|>t\},Y_{\even}^{(m+1)}),
\end{aligned}
\end{equation}
where $t=0,1,\dots,m.$
For $m=2$ and $t=0,1,2$ in Fig.~\ref{fig:fceng}, we have
$Y_{2,\odd}^{(3)} = Y_{\odd}^{(3)}, 
    Y_{2,\even}^{(3)} = Y_{\even}^{(3)},  
    Y_{1,\odd}^{(3)} = (U_{\{3,2,1\}}^{(3)}+U_{\{2,1\}}^{(3)},Y_{\odd}^{(3)}), 
    Y_{1,\even}^{(3)} = (U_{\{2,1\}}^{(3)},Y_{\even}^{(3)}),  
Y_{0,\odd}^{(3)} = (U_{\{3,2,1\}}^{(3)}+U_{\{2,1\}}^{(3)},U_{\{3,1\}}^{(3)}+U_{\{1\}}^{(3)},U_{\{3,2\}}^{(3)}+U_{\{2\}}^{(3)},Y_{\odd}^{(3)}), 
Y_{0,\even}^{(3)} = (U_{\{2,1\}}^{(3)},U_{\{1\}}^{(3)},U_{\{2\}}^{(3)},Y_{\even}^{(3)})$.

The bit-channels $\{W_A^{(2)}:A\subseteq[2]\}$ are divided into three layers according to the cardinality of the set $A$: The $0$-th layer is $\{W_{\emptyset}^{(2)}\}$, the only set with cardinality $0$;
the first layer is $\{W_{\{1\}}^{(2)}, W_{\{2\}}^{(2)}\}$, the sets with cardinality $1$; and the second layer is $\{W_{\{2,1\}}^{(2)}\}$, the only set with cardinality $2$;
In Fig.~\ref{fig:fceng}, we can see that the bit-channels in the next level $\{W_A^{(3)}:A\subseteq[3]\}$ are obtained by taking two copies of each layer in $\{W_A^{(2)}:A\subseteq[2]\}$ and performing the fast polar transform. More precisely, for the $0$-th and the second layer, we perform the standard polar transform, and for the first layer, we perform the fast polar transform with block size $2$. We will discuss the fast polar transform in more detail later in this section. 

Such a polarization procedure in fact describes the recursive structure between the bit-channels $\{W_A^{(m)}:A\subseteq[m]\}$ in the $m$-th level and the bit-channels $\{W_A^{(m+1)}:A\subseteq[m+1]\}$ in the $(m+1)$-th level for general values of $m$. More specifically, The bit-channels $\{W_A^{(m)}:A\subseteq[m]\}$ are divided into $m+1$ layers according to the cardinality of the set $A$: The $i$-th layer is $\{W_A^{(m)}:A\subseteq[m],|A|=i\}$, the sets with cardinality $i$, for $i=0,1,2,\dots,m$.
Then we take two copies of each layer $\{W_A^{(m)}:A\subseteq[m],|A|=i\}$ and perform the fast polar transform. The outcome of the ``$-$" fast polar transform is the bit-channels $\{W_{A\cup\{m+1\}}^{(m+1)}:A\subseteq[m],|A|=i\}$ in the next level, and the outcome of the ``$+$" fast polar transform is the bit-channels $\{W_A^{(m+1)}:A\subseteq[m],|A|=i\}$.
From the perspective of the bit-channels $\{W_A^{(m+1)}:A\subseteq[m+1]\}$ in the $(m+1)$-th level, except for the $0$-th and the $(m+1)$-th layers, each layer $\{W_A^{(m+1)}:A\subseteq[m+1],|A|=i\}$ is divided into two parts: The first part is $\{W_A^{(m+1)}:A\subseteq[m],|A|=i\}$, which is the ``$+$" fast polar transform of the $i$-th layer $\{W_A^{(m)}:A\subseteq[m],|A|=i\}$ in the $m$-th level. The second part is $\{W_{A\cup\{m+1\}}^{(m+1)}:A\subseteq[m],|A|=i-1\}$, which is the ``$-$" fast polar transform of the $(i-1)$-th layer $\{W_A^{(m)}:A\subseteq[m],|A|=i-1\}$ in the $m$-th level. 
As for the $0$-th and the $(m+1)$-th layers, each of them only contains a single bit-channel $W_{\emptyset}^{(m+1)}$ and $W_{[m+1]}^{(m+1)}$, respectively, where $W_{\emptyset}^{(m+1)}$ is the ``$+$" polar transform of $W_{\emptyset}^{(m)}$, and $W_{[m+1]}^{(m+1)}$ is the ``$-$" polar transform of $W_{[m]}^{(m)}$. See Fig.~\ref{fig:rs} for an illustration of $m=4$.

Next we explain the fast polar transform. Let us consider the bit-channels $\{W_A^{(m+1)}:A\subseteq[m+1]\}$ in the $(m+1)$-th level. According to the total order defined in Definition~\ref{def:tto}, the layers from top to down are the $(m+1)$-th layer, the $m$-th layer, \dots, all the way down to the $0$-th layer. Within each layer, all the sets containing the element $m+1$ appear after those not containing this element. Now let $B_1<B_2<\dots<B_j$ be all the sets in the $t$-th layer of $[m]$, i.e., they are all the subsets of $[m]$ with cardinality $t$. Then by the discussion above we know that the following sequence of subsets
\begin{equation} \label{eq:spo}
B_1\cup\{m+1\}<B_2\cup\{m+1\}<\dots<B_j\cup\{m+1\}<B_1<B_2<\dots<B_j
\end{equation}
are consecutive according to the total order on the subsets of $[m+1]$.
It is easy to check that given $Y_{t,\odd}^{(m+1)}$ and $Y_{t,\even}^{(m+1)}$ (see \eqref{eq:toe} for definitions), if a successive decoder decodes according to the order in \eqref{eq:spo}, then the bit-channels seen by this decoder are equivalent to
$$
W_{B_1\cup\{m+1\}}^{(m+1)}, W_{B_2\cup\{m+1\}}^{(m+1)}, \dots, W_{B_j\cup\{m+1\}}^{(m+1)},
W_{B_1}^{(m+1)}, W_{B_2}^{(m+1)}, \dots, W_{B_j}^{(m+1)}.
$$
For instance, $W_{B_2\cup\{m+1\}}^{(m+1)}$ is equivalent to the bit-channel mapping from $U_{B_2\cup\{m+1\}}^{(m+1)}$ to $(U_{B_1\cup\{m+1\}}^{(m+1)}, \linebreak[4]
Y_{t,\odd}^{(m+1)},  Y_{t,\even}^{(m+1)})$.
This is illustrated in Fig.~\ref{fig:bpt}.

In order to connect fast polar transform to the standard polar transform, we consider the following order of the sets in \eqref{eq:spo}:
\begin{equation} \label{eq:ljk}
B_1\cup\{m+1\}, B_1, B_2\cup\{m+1\}, B_2,\dots,B_j\cup\{m+1\},B_j.
\end{equation}
Assuming that we are still given $Y_{t,\odd}^{(m+1)}$ and $Y_{t,\even}^{(m+1)}$, but this time the successive decoder decodes in this order instead of the order in \eqref{eq:spo}. We denote the bit-channels seen by this successive decoder as 
$$
\widetilde{W}_{B_1\cup\{m+1\}}^{(m+1)}, \widetilde{W}_{B_1}^{(m+1)}, \widetilde{W}_{B_2\cup\{m+1\}}^{(m+1)}, \widetilde{W}_{B_2}^{(m+1)}, \dots, \widetilde{W}_{B_j\cup\{m+1\}}^{(m+1)}, \widetilde{W}_{B_j}^{(m+1)}.
$$
We further define
\begin{equation} \label{eq:lby}
\begin{aligned}
    Y_{B_i,\odd}^{(m+1)} &:= (\{U_{B\cup\{m+1\}}^{(m+1)}+U_B^{(m+1)}:B\subseteq[m],B<B_i\},Y_{\odd}^{(m+1)}), \\
    Y_{B_i,\even}^{(m+1)} &:= (\{U_B^{(m+1)}:B\subseteq[m],B<B_i\},Y_{\even}^{(m+1)}).
\end{aligned}
\end{equation}

According to Fig.~\ref{fig:1and2}, $\widetilde{W}_{B_i\cup\{m+1\}}^{(m+1)}$ and $\widetilde{W}_{B_i}^{(m+1)}$ are ``$-$" and ``$+$" polar transforms of $W_{B_i}^{(m)}$, respectively.
Then by \eqref{eq:kbd} and Lemma~\ref{lm:sas}, we know that 
\begin{equation} \label{eq:cdb1}
\widetilde{H}_{B_i\cup\{m+1\}}^{(m+1)} \ge H_{B_i}^{(m)} \ge \widetilde{H}_{B_i}^{(m+1)},
\end{equation}
and if $H_{B_i}^{(m)}\in(\epsilon, 1-\epsilon)$, then
\begin{equation} \label{eq:cdb2}
\widetilde{H}_{B_i\cup\{m+1\}}^{(m+1)} - H_{B_i}^{(m)} > \delta
\text{~~and~~} H_{B_i}^{(m)} - \widetilde{H}_{B_i}^{(m+1)} > \delta,
\end{equation}
where $\widetilde{H}=1-I(\widetilde{W})$.
Comparing the order in \eqref{eq:spo} and \eqref{eq:ljk}, we can see that every set that appears before $B_i\cup\{m+1\}$ in \eqref{eq:spo} also appears before $B_i\cup\{m+1\}$ in \eqref{eq:ljk}.
Therefore, for every $i\in\{1,2,\dots,j\}$, we have $H_{B_i\cup\{m+1\}}^{(m+1)}\ge \widetilde{H}_{B_i\cup\{m+1\}}^{(m+1)}$. Similarly, we also have $\widetilde{H}_{B_i}^{(m+1)}\ge H_{B_i}^{(m+1)}$. In other words, in the standard polar transform, we obtain a worse bit-channel through the ``$-$" transform and a better one through the ``$+$" transform. Then in the fast polar transform, we make the bit-channel obtained through the standard ``$-$" polar transform even worse and the bit-channel obtained through the standard ``$+$" polar transform even better. Therefore, the gap between $H_{B_i\cup\{m+1\}}^{(m+1)}$ and $H_{B_i}^{(m)}$ is even larger than the gap between $\widetilde{H}_{B_i\cup\{m+1\}}^{(m+1)}$ and $H_{B_i}^{(m)}$. Similarly, the gap between $H_{B_i}^{(m)}$ and $H_{B_i}^{(m+1)}$ is even larger than the gap between $H_{B_i}^{(m)}$ and $\widetilde{H}_{B_i}^{(m+1)}$; see Fig.~\ref{fig:gp} for an illustration. Combining this with \eqref{eq:cdb1}--\eqref{eq:cdb2}, we have shown that \eqref{eq:jjw1}--\eqref{eq:jjw2} hold for any $A_{i+1}=A_i\cup\{m+1\}$.

By now, we have explained how to prove Theorem~\ref{thm:m1} and Theorem~\ref{thm:m2}. The next step is to use these two theorems to prove Theorem~\ref{thm:m3}. To that end, we need the following strengthened form of Theorem~\ref{thm:m2}.

\begin{theorem}\label{thm:strong2}
For every BMS channel $W$ and every $0<\epsilon<0.1$, any $\delta_n=\poly(1/n)$ and $0<\gamma<1$, there is a constant $M(\epsilon,\delta_n,\gamma)$ such that for every $m>M(\epsilon,\delta_n,\gamma)$ and every increasing chain of sets $\emptyset=A_0 \subseteq A_1 \subseteq A_2 \subseteq \dots \subseteq A_m=[m]$,
\begin{equation} \label{eq:ccj}
\left|\left\{ i\in\{0,1,\dots,m\}:H_{A_i}^{(m)} > 1-\epsilon \right\} \cup
\left\{ i\in\{0,1,\dots,m\}:Z_{A_i}^{(m)} < \delta_n \right\} \right| \ge m-m^\gamma. 
\end{equation}
\end{theorem}

The proof of this theorem mainly relies on the fact that the Bhattacharyya parameter is close to $0$ if and only if the conditional entropy is close to $0$. More precisely, the proof relies on the following two well-known inequalities
in the polar coding literature (see Proposition 1 of \cite{Arikan09} for a proof):
\begin{align}
Z(X|Y) &\ge H(X|Y), \label{eq:je} \\
(1-H(X|Y))^2 &\le 1- (Z(X|Y))^2. \label{eq:os}
\end{align}
We switch from the conditional entropy in Theorem~\ref{thm:m2} to the Bhattacharyya parameter in Theorem~\ref{thm:strong2} and Theorem~\ref{thm:m3} for two reasons: First, the Bhattacharyya parameter $Z(X|Y)$ is an upper bound on the error probability of
the MAP decoder of $X$ given $Y$, i.e.\ (see \cite{Arikan09}),
\begin{equation}\label{eq:pez}
P_e(X|Y) \le Z(X|Y).
\end{equation}
This property makes it convenient for us to prove that the twin codes achieve capacity (Theorem~\ref{thm:cac}). Second, in the ``$+$" polar transform, the evolution of Bhattacharyya parameters follows a square law \eqref{eq:ks}. As a result, it is easier to obtain a better bound on the Bhattacharyya parameters than on the conditional entropy.

Once we prove Theorem~\ref{thm:strong2}, we  further use that there are $m!$ distinct increasing chains of sets for a given $m$.
Let us fix $m$ and list all the $m!$ distinct increasing chains of sets as follows:
\begin{align*}
\emptyset=A_0(1) \subseteq A_1(1) \subseteq A_2(1) \subseteq & \dots \subseteq A_m(1)=[m], \\
\emptyset=A_0(2) \subseteq A_1(2) \subseteq A_2(2) \subseteq & \dots \subseteq A_m(2)=[m], \\
\emptyset=A_0(3) \subseteq A_1(3) \subseteq A_2(3) \subseteq & \dots \subseteq A_m(3)=[m], \\
\vdots \hspace*{0.2in} \vdots \hspace*{0.2in} \vdots \hspace*{0.2in} & \vdots \\
\emptyset=A_0(m!) \subseteq A_1(m!) \subseteq A_2(m!) \subseteq & \dots \subseteq A_m(m!)=[m].
\end{align*}
In Theorem~\ref{thm:strong2}, we have shown that among each increasing chain of sets, almost all the bit-channels becomes either completely noisy or noiseless.
Let $\cA$ be the collection of all the ``bad" subsets of $[m]$, and let $\cS$ be the collection of all the ``bad" sets in the above $m!$ increasing chains (including multiplicity), where ``bad" means that the set does not belong to the left-hand side of \eqref{eq:ccj}.
Theorem~\ref{thm:strong2} tells us that the ``bad" sets in each chain is upper bounded by $m^{\gamma}$, so $|\cS|\le m^{\gamma} m!$. On the other hand, notice that each subset with cardinality $i$ appears $i!(m-i)!$ times in all the $m!$ increasing chains listed above, and that $i!(m-i)! \ge \lfloor m/2 \rfloor!(m-\lfloor m/2 \rfloor)!$ for all $i\in[m]$. Therefore, 
$|\cS|\ge \lfloor m/2 \rfloor!(m-\lfloor m/2 \rfloor)!|\cA|$. Combining the upper and lower bounds of $|\cS|$, we obtain an upper bound on $|\cA|$, and this proves Theorem~\ref{thm:m3}. Theorem~\ref{thm:cac} in turn follows directly from Theorem~\ref{thm:m3} (using  \eqref{eq:pez}).

In the next section, we fill out all the details and give full proofs of the main theorems. We first prove two technical lemmas that play essential roles in most of the proofs. More precisely, Lemma~\ref{lm:is} allows us to compare $H_A^{(m)}$ and $H_B^{(m)}$ when $|A|=|B|$, and its proof relies on the symmetry of RM codes. Lemma~\ref{lm:imp} formalizes \eqref{eq:jjw1}--\eqref{eq:jjw2} and also states their counterparts for Bhattacharyya parameters. Then the proofs of Theorem~\ref{thm:strong2} and all the theorems in Section~\ref{sect:mm} are given in the last section.

\section{Conclusion}
Recently, Hassani et al. gave theoretical results backing the conjecture that RM codes have an almost optimal scaling law over BSC channels under ML decoding \cite{Hassani18}, where optimal scaling law means that for a fixed linear code, the decoding error probability of ML decoder transitions from $0$ to $1$ as a function of the crossover probability of the BSC channel in the sharpest manner (i.e., comparable to random codes).
In this paper, we have demonstrated that RM codes polarize faster than polar codes --- see Fig.~\ref{fig:gp} --- even though we stated our bound in Theorem \ref{thm:m3} by exploiting the polar code bounds.   
A future research direction would thus be to use directly the fast polarization of RM codes to prove that RM codes have a better scaling law than polar codes and/or to prove that RM codes have an ``optimal'' scaling law.

Finally, this paper gives a second ordering of the matrix $G_m$ that polarizes, i.e., the RM coder ordering in addition to the polar code ordering (and the various other equivalent orderings that result from both of these). Are there many more\footnote{Clearly some ordering do not polarize, such as the down-to-top ordering in $G_m$.} orderings that polarize? Is the RM code ordering ``optimal''?

\section{Proofs}
\subsection{Two technical lemmas}
We first need to establish some symmetry properties of RM codes and their impact on the conditional mutual information. Denote by $S_m$ the symmetric group of order $m$. 
For $\pi\in S_m$ and $A\subseteq[m]$, define
$\pi(A) := \{ \pi(a):a\in A \}$.
Note that $S_m$ is contained in the automorphism group of RM codes, as any degree $\le k$ polynomial is a degree $\le k$ polynomial under a relabelling of its variables.

Let $A\subseteq[m]$, and let $\cB$ be a subset of the power set of $[m]$.
For any $\pi\in S_m$, i.e., any  relabelling of the elements of $[m]$, we have 
\begin{equation} \label{eq:sym1}
    H(U_A^{(m)}|Y^{(m)},\{U_B^{(m)}:B\in\cB\})
= H(U_{\pi(A)}^{(m)}|Y^{(m)},\{U_{\pi(B)}^{(m)}:B\in\cB\}).
\end{equation}
This equality leads to the following lemma.

\begin{lemma}\label{lm:is}
Let $W$ be a BMS channel.
Let $A\subset[m]$ and $i_1,i_2\in[m]$ satisfy that $i_1,i_2\notin A$ and $i_1<i_2$. Then
$$
H_{A\cup\{i_1\}}^{(m)}  \ge  H_{A\cup\{i_2\}}^{(m)}
\quad \text{~and~} \quad 
Z_{A\cup\{i_1\}}^{(m)}  \ge  Z_{A\cup\{i_2\}}^{(m)}.
$$
\end{lemma}
\begin{proof}
Define $\pi\in S_m$ as
\begin{equation}\label{eq:pi}
\pi(i)=i \text{~for all~} i\neq i_1,i_2, \quad
\pi(i_1)=i_2, \quad \pi(i_2)=i_1.
\end{equation}
By \eqref{eq:sym1}, we have
\begin{align*}
H_{A\cup\{i_1\}}^{(m)} & = 
H \Big( U_{A\cup\{i_1\}}^{(m)} \Big| Y^{(m)},\{U_B^{(m)}:B\subseteq[m],B<(A\cup\{i_1\})\} \Big) \\
& = H \Big( U_{\pi(A\cup\{i_1\})}^{(m)} \Big| Y^{(m)},
\{U_{\pi(B)}^{(m)}:B\subseteq[m],B<(A\cup\{i_1\})\} \Big) \\
& = H \Big( U_{A\cup\{i_2\}}^{(m)} \Big| Y^{(m)},
\{U_{\pi(B)}^{(m)}:B\subseteq[m],B<(A\cup\{i_1\})\} \Big) \\
& \ge H \Big( U_{A\cup\{i_2\}}^{(m)} \Big| Y^{(m)},
\{U_B^{(m)}:B\subseteq[m],B<(A\cup\{i_2\})\} \Big) \\
& = H_{A\cup\{i_2\}}^{(m)},
\end{align*}
where the inequality follows from the fact that
\begin{equation}\label{eq:sct}
\{\pi(B): B\subseteq[m],B<(A\cup\{i_1\})\} \subseteq \{B: B\subseteq[m],B<(A\cup\{i_2\}) \}.
\end{equation}
Indeed, if $B<(A\cup\{i_1\})$ and $i_1\notin B$, then $\pi(B)\le B<(A\cup\{i_1\})<(A\cup\{i_2\})$.
If $B<(A\cup\{i_1\})$ and $i_1\in B$, then $(B\setminus\{i_1\})<A$, so 
$$
\pi(B)=\pi((B\setminus\{i_1\})\cup\{i_1\}) = \pi(B\setminus\{i_1\})\cup\{i_2\}\le (B\setminus\{i_1\})\cup\{i_2\}<A\cup\{i_2\}.
$$
Therefore we have shown that $B<(A\cup\{i_1\})$ implies $\pi(B)<(A\cup\{i_2\})$, which is exactly the set containment in \eqref{eq:sct}. This completes the proof of the lemma.
Using Lemma~\ref{lm:ox} and the same reasoning as above, one can easily show that $Z_{A\cup\{i_1\}}^{(m)}  \ge  Z_{A\cup\{i_2\}}^{(m)}$.
\end{proof}

\begin{lemma}\label{lm:imp}
For every BMS channel $W$, every positive integer $m$, every $A\subseteq[m]$ and every $j\in[m+1]\setminus A$, we have the interlacing property:
\begin{align}
H_{A\cup\{j\}}^{(m+1)} & \ge H_A^{(m)} \ge H_A^{(m+1)} ,  \label{eq:jw}  \\
Z_A^{(m+1)} & \le \left( Z_A^{(m)} \right)^2 , \quad \quad
 Z_A^{(m)} \le Z_{A\cup\{j\}}^{(m+1)}.
\label{eq:qqt}
\end{align}
Moreover, for any $\epsilon>0$, there is $\delta(\epsilon)>0$ such that
for any positive integer $m$, any $A\subseteq[m]$ and any $j\in[m+1]\setminus A$,
$$
H_A^{(m)} \in (\epsilon,1-\epsilon)
$$
implies that
\begin{equation} \label{eq:jf}
 H_A^{(m)} -  H_A^{(m+1)} > \delta(\epsilon)
\text{~and~} 
 H_{A\cup\{j\}}^{(m+1)} - H_A^{(m)} > \delta(\epsilon).
\end{equation}
\end{lemma}
\begin{proof}
Recall the definition of $Y_{\odd}^{(m+1)}$ and $Y_{\even}^{(m+1)}$ in \eqref{eq:oddeven}.
Let $y^{(m)}=(y_{\mathbi{z}}^{(m)}: \mathbi{z}\in\{0,1\}^m) \in \cY^n$ be a vector of length $n=2^m$ whose components take values in $\cY$, and this vector is indexed by $\mathbi{z}\in\{0,1\}^m$, which is similar to the random vector $Y^{(m)}$.
Let $u^{(m)}=(u_A^{(m)}: A\subseteq[m]) \in \{0,1\}^n$ be a binary vector of length $n=2^m$, and this vector is indexed by $A\subseteq[m]$, which is similar to the random vector $U^{(m)}$.
For $y^{(m)}\in \cY^n$, define the following three events:
\begin{align*}
\{Y^{(m)} = y^{(m)}\} & := \{Y_{\mathbi{z}}^{(m)} = y_{\mathbi{z}}^{(m)} \text{~for all~} 
\mathbi{z}\in\{0,1\}^m\},  \\
\{Y_{\odd}^{(m+1)} = y^{(m)}\} & := \{Y_{(\mathbi{z},1)}^{(m+1)} = y_{\mathbi{z}}^{(m)} \text{~for all~}  \mathbi{z}\in\{0,1\}^m\}, \\
\{Y_{\even}^{(m+1)} = y^{(m)}\} & := \{Y_{(\mathbi{z},0)}^{(m+1)} = y_{\mathbi{z}}^{(m)} \text{~for all~}  \mathbi{z}\in\{0,1\}^m\},
\end{align*}
where for $\mathbi{z}=(z_1,z_2,\dots,z_m)\in\{0,1\}^m$, 
$(\mathbi{z},1):=(z_1,z_2,\dots,z_m,1)$ and $(\mathbi{z},0):=(z_1,z_2,\dots,z_m,0)$.
The main observation is that for any $y^{(m)}\in \cY^n$ and any $u^{(m)}=(u_A^{(m)}: A\subseteq[m])  \in \{0,1\}^n$,
\begin{equation}\label{eq:ej}
\begin{aligned}
& P \Big( \{Y_{\odd}^{(m+1)}=y^{(m)}\} \Big| \{U_{A}^{(m+1)}+U_{A\cup\{m+1\}}^{(m+1)}=u_A^{(m)} \text{~for all~} A\subseteq[m]\} \Big) \\
= & P \Big( \{Y_{\even}^{(m+1)}=y^{(m)}\} \Big| \{U_{A}^{(m+1)}=u_A^{(m)} \text{~for all~} A\subseteq[m]\} \Big) \\
= & P \Big( \{Y^{(m)}=y^{(m)}\} \Big| \{U_{A}^{(m)}=u_A^{(m)} \text{~for all~} A\subseteq[m]\} \Big).
\end{aligned}
\end{equation} 
Since the two vectors $(U_{A}^{(m+1)}+U_{A\cup\{m+1\}}^{(m+1)}:A\subseteq[m])$
and $(U_{A}^{(m+1)}:A\subseteq[m])$ are independent,
$(Y_{\odd}^{(m+1)},\{U_{A}^{(m+1)}+U_{A\cup\{m+1\}}^{(m+1)}:A\subseteq[m]\})$ and 
$(Y_{\even}^{(m+1)},\{U_{A}^{(m+1)}:A\subseteq[m]\})$ are also independent.
By \eqref{eq:ej}, we can also obtain that
\begin{align}
H_A^{(m)} &= H \Big(U_A^{(m)} \Big| Y^{(m)}, U_{<A}^{(m)} \Big)
=H\Big(U_A^{(m+1)} \Big|Y_{\even}^{(m+1)},\{U_{A'}^{(m+1)}:A'\subseteq[m],A'<A\} \Big)
\nonumber\\
&=H \Big(U_{A}^{(m+1)}+U_{A\cup\{m+1\}}^{(m+1)} \Big|Y_{\odd}^{(m+1)},
\{U_{A'}^{(m+1)}+U_{A'\cup\{m+1\}}^{(m+1)}:A'\subseteq[m],A'<A\} \Big). \label{eq:pw}
\end{align}
Therefore, for any $A\subseteq[m]$,
\begin{equation}\label{eq:mtg}
\begin{aligned}
& H\Big(U_A^{(m+1)} \Big|U_{A\cup\{m+1\}}^{(m+1)}, Y^{(m+1)}, 
\{U_{A'}^{(m+1)}:A'\subseteq[m],A'<A\},
\{U_{A'\cup\{m+1\}}^{(m+1)}:A'\subseteq[m],A'<A\} \Big) \\
&+ H\Big(U_{A\cup\{m+1\}}^{(m+1)} \Big| Y^{(m+1)}, 
\{U_{A'}^{(m+1)}:A'\subseteq[m],A'<A\},
\{U_{A'\cup\{m+1\}}^{(m+1)}:A'\subseteq[m],A'<A\} \Big)\\
= & H\Big(U_A^{(m+1)},U_{A\cup\{m+1\}}^{(m+1)} \Big| Y^{(m+1)}, 
\{U_{A'}^{(m+1)}:A'\subseteq[m],A'<A\},
\{U_{A'\cup\{m+1\}}^{(m+1)}:A'\subseteq[m],A'<A\} \Big) \\
= & H\Big(U_A^{(m+1)},U_A^{(m+1)}+U_{A\cup\{m+1\}}^{(m+1)} \Big| Y_{\even}^{(m+1)}, 
\{U_{A'}^{(m+1)}:A'\subseteq[m],A'<A\}, \\
& \hspace*{2.1in} Y_{\odd}^{(m+1)},
\{U_{A'}^{(m+1)}+U_{A'\cup\{m+1\}}^{(m+1)}:A'\subseteq[m],A'<A\} \Big) \\
\overset{(a)}{=} & H\Big(U_A^{(m+1)} \Big| Y_{\even}^{(m+1)}, 
\{U_{A'}^{(m+1)}:A'\subseteq[m],A'<A\} \Big) \\
& + H\Big(U_A^{(m+1)}+U_{A\cup\{m+1\}}^{(m+1)} \Big|  Y_{\odd}^{(m+1)},
\{U_{A'}^{(m+1)}+U_{A'\cup\{m+1\}}^{(m+1)}:A'\subseteq[m],A'<A\} \Big) \\
=& 2H_A^{(m)},
\end{aligned}
\end{equation}
where equality $(a)$ holds because $(Y_{\odd}^{(m+1)},\{U_{A}^{(m+1)}+U_{A\cup\{m+1\}}^{(m+1)}:A\subseteq[m]\})$ and 
$(Y_{\even}^{(m+1)},\{U_{A}^{(m+1)}:A\subseteq[m]\})$ are independent.
It is also clear that
\begin{align*}
& H\Big(U_A^{(m+1)} \Big|U_{A\cup\{m+1\}}^{(m+1)}, Y^{(m+1)}, 
\{U_{A'}^{(m+1)}:A'\subseteq[m],A'<A\},
\{U_{A'\cup\{m+1\}}^{(m+1)}:A'\subseteq[m],A'<A\} \Big) \\
\le & H\Big(U_A^{(m+1)} \Big| Y_{\even}^{(m+1)}, 
\{U_{A'}^{(m+1)}:A'\subseteq[m],A'<A\} \Big) = H_A^{(m)},
\end{align*}
so we have
\begin{equation}\label{eq:yf}
\begin{aligned}
 H\Big(U_A^{(m+1)} \Big|U_{A\cup\{m+1\}}^{(m+1)}, Y^{(m+1)}, 
\{U_{A'}^{(m+1)}:A'\subseteq[m],A'<A\},
\{U_{A'\cup\{m+1\}}^{(m+1)}:A'\subseteq[m],A'<A\} \Big) \\
\le   H_A^{(m)}
\le  H\Big(U_{A\cup\{m+1\}}^{(m+1)} \Big| Y^{(m+1)}, 
\{U_{A'}^{(m+1)}:A'\subseteq[m],A'<A\},
\{U_{A'\cup\{m+1\}}^{(m+1)}:A'\subseteq[m],A'<A\} \Big).
\end{aligned}
\end{equation}
According to the ordering of sets defined in the previous section, it is easy to verify the following relation:
\begin{align}
& \Big(\{A\cup\{m+1\}\} \cup\{A':A'\subseteq[m],A'<A\}\cup
\{A'\cup\{m+1\}:A'\subseteq[m],A'<A\} \Big) \nonumber \\
 \subseteq & \{A':A'\subseteq[m+1],A'<A\}; \label{eq:if1} \\
& \{A':A'\subseteq[m+1],A'<(A\cup\{m+1\} ) \} \nonumber \\
\subseteq & \{A':A'\subseteq[m],A'<A\} \cup
\{A'\cup\{m+1\}:A'\subseteq[m],A'<A\}. \label{eq:if2}
\end{align}
As a consequence,
\begin{equation}\label{eq:pd}
\begin{aligned}
& H_{A}^{(m+1)} \\
& \le  H\Big(U_A^{(m+1)} \Big|U_{A\cup\{m+1\}}^{(m+1)}, Y^{(m+1)}, 
\{U_{A'}^{(m+1)}:A'\subseteq[m],A'<A\},
\{U_{A'\cup\{m+1\}}^{(m+1)}:A'\subseteq[m],A'<A\} \Big), \\
& H_{A\cup\{m+1\}}^{(m+1)} \\
& \ge
H\Big(U_{A\cup\{m+1\}}^{(m+1)} \Big| Y^{(m+1)}, 
\{U_{A'}^{(m+1)}:A'\subseteq[m],A'<A\},
\{U_{A'\cup\{m+1\}}^{(m+1)}:A'\subseteq[m],A'<A\} \Big).
\end{aligned}
\end{equation}
Combining these two inequalities with \eqref{eq:yf},
we have $H_{A\cup\{m+1\}}^{(m+1)} \ge H_A^{(m)} \ge H_A^{(m+1)}$. Then by Lemma~\ref{lm:is}, for any $j\in[m+1]\setminus A$, we have $H_{A\cup\{j\}}^{(m+1)} \ge H_{A\cup\{m+1\}}^{(m+1)} \ge H_A^{(m)} \ge H_A^{(m+1)}$.
This completes the proof of \eqref{eq:jw}. 

Next we prove \eqref{eq:qqt}.
Let 
\begin{align*}
& X_1:=U_A^{(m+1)} , \quad \quad
Y_1:= \Big(Y_{\even}^{(m+1)}, 
\{U_{A'}^{(m+1)}:A'\subseteq[m],A'<A\} \Big), \\
& X_2:= U_A^{(m+1)}+U_{A\cup\{m+1\}}^{(m+1)} , \quad\quad
Y_2:=  \Big( Y_{\odd}^{(m+1)},
\{U_{A'}^{(m+1)}+U_{A'\cup\{m+1\}}^{(m+1)}:A'\subseteq[m],A'<A\} \Big), \\
& X:=U_A^{(m)} , \quad \quad
Y:= \Big(Y^{(m)}, 
\{U_{A'}^{(m)}:A'\subseteq[m],A'<A\} \Big).
\end{align*}
Then $(X_1,Y_1)$ and $(X_2,Y_2)$ are i.i.d., and they have the same distribution as $(X,Y)$.
By \eqref{eq:ks} we have
\begin{align*}
& Z \left( U_A^{(m+1)} \Big| Y^{(m+1)},\{U_{A'}^{(m+1)}:A'\subseteq[m],A'<A\},
\{U_{A'\cup\{m+1\}}^{(m+1)}:A'\subseteq[m],A'<A\},U_{A\cup\{m+1\}}^{(m+1)} \right) \\
& = Z(X_1|Y_1,Y_2,X_1+X_2)
=(Z(X|Y))^2 = \left( Z_A^{(m)} \right)^2.
\end{align*}
According to \eqref{eq:if1} and Lemma~\ref{lm:ox},
\begin{equation} \label{eq:atl}
\begin{aligned}
& Z_A^{(m+1)} = Z(U_A^{(m+1)}|Y^{(m+1)},U_{<A}^{(m+1)}) \\
& \le Z(U_A^{(m+1)}| Y^{(m+1)},\{U_{A'}^{(m+1)}:A'\subseteq[m],A'<A\},
\{U_{A'\cup\{m+1\}}^{(m+1)}:A'\subseteq[m],A'<A\},U_{A\cup\{m+1\}}^{(m+1)}) \\
& = \left( Z_A^{(m)} \right)^2.
\end{aligned}
\end{equation}
By \eqref{eq:lsl}, we have
\begin{align*}
& Z \left( U_{A\cup\{m+1\}}^{(m+1)} \Big| Y^{(m+1)},\{U_{A'}^{(m+1)}:A'\subseteq[m],A'<A\},
\{U_{A'\cup\{m+1\}}^{(m+1)}:A'\subseteq[m],A'<A\} \right) \\
= & Z(X_1+X_2|Y_1,Y_2) \ge Z(X|Y) = Z_A^{(m)}.
\end{align*}
Combining this with \eqref{eq:if2} and Lemma~\ref{lm:ox}, we obtain that
\begin{align*}
& Z_{A\cup\{m+1\}}^{(m+1)} = Z(U_{A\cup\{m+1\}}^{(m+1)}|Y^{(m+1)},U_{<(A\cup\{m+1\})}^{(m+1)}) \\
& \ge Z \left( U_{A\cup\{m+1\}}^{(m+1)} \Big| Y^{(m+1)},\{U_{A'}^{(m+1)}:A'\subseteq[m],A'<A\},
\{U_{A'\cup\{m+1\}}^{(m+1)}:A'\subseteq[m],A'<A\} \right) \\
& \ge Z_A^{(m)}.
\end{align*}
By Lemma~\ref{lm:is}, for any $j\in[m+1]\setminus A$, we further have that $Z_{A\cup\{j\}}^{(m+1)} \ge Z_{A\cup\{m+1\}}^{(m+1)} \ge Z_A^{(m)}$.
Combining this with \eqref{eq:atl}, we complete the proof of \eqref{eq:qqt}.

Now we prove \eqref{eq:jf}. 
For every $\epsilon>0$, we use the same $\delta(\epsilon)>0$ as in Lemma~\ref{lm:sas}.
We assume  
$H_A^{(m)} \in (\epsilon,1-\epsilon)$ and use Lemma~\ref{lm:sas} to prove \eqref{eq:jf} under this assumption.
Since $(X_1,Y_1)$ and $(X_2,Y_2)$ are i.i.d. with the same distribution as $(X,Y)$, we have
$$
H(X_1|Y_1)=H(X_2|Y_2)=H(X|Y)=H_A^{(m)} \in (\epsilon,1-\epsilon).
$$
According to Lemma~\ref{lm:sas},
\begin{equation}\label{eq:fw}
H(X_1+X_2|Y_1,Y_2) \ge H_A^{(m)} + \delta(\epsilon).
\end{equation}
Also observe that 
\begin{align*}
& H(X_1+X_2|Y_1,Y_2)  \\
=& H\Big(U_{A\cup\{m+1\}}^{(m+1)} \Big| Y^{(m+1)}, 
\{U_{A'}^{(m+1)}:A'\subseteq[m],A'<A\},
\{U_{A'\cup\{m+1\}}^{(m+1)}:A'\subseteq[m],A'<A\} \Big).
\end{align*}
Therefore by \eqref{eq:mtg} and \eqref{eq:fw}, we have
\begin{align*}
& H\Big(U_{A\cup\{m+1\}}^{(m+1)} \Big| Y^{(m+1)}, 
\{U_{A'}^{(m+1)}:A'\subseteq[m],A'<A\},
\{U_{A'\cup\{m+1\}}^{(m+1)}:A'\subseteq[m],A'<A\} \Big) - H_A^{(m)} \\
= & H_A^{(m)}- H\Big(U_A^{(m+1)} \Big|U_{A\cup\{m+1\}}^{(m+1)}, Y^{(m+1)}, 
\{U_{A'}^{(m+1)}:A'\subseteq[m],A'<A\},
\{U_{A'\cup\{m+1\}}^{(m+1)}:A'\subseteq[m],A'<A\} \Big) \\
\ge & \delta(\epsilon).
\end{align*}
Combining this with \eqref{eq:pd} and Lemma~\ref{lm:is}, we conclude that for any $j\in[m+1]\setminus A$,
$$
 H_A^{(m)} -  H_A^{(m+1)} > \delta(\epsilon)
\text{~and~} 
 H_{A\cup\{j\}}^{(m+1)} - H_A^{(m)} \ge H_{A\cup\{m+1\}}^{(m+1)} - H_A^{(m)} > \delta(\epsilon).
$$
This completes the proof of the lemma.
\end{proof}

\begin{lemma}[Lemma 1.8 in \cite{Korada09}] \label{lm:ox}
Let $(X,Y,Y')$ be a triple of discrete random variables, where $X$ has Bernoulli-$1/2$ distribution. Then
$$
Z(X|Y,Y') \le Z(X|Y).
$$
\end{lemma}
This lemma can be proved by a straightforward application of the Cauchy-Schwarz inequality.

\subsection{Proof of Theorem~\ref{thm:strong2}}
Without loss of generality, assume that $\delta_n=n^{-d}$ for some positive constant $d$.
If $A\subseteq B\subseteq [m]$ and $|B|=|A|+1$, then by \eqref{eq:qqt},
\begin{equation}\label{eq:AB}
Z_{A}^{(m)} \le Z_{B}^{(m+1)} \le \left( Z_{B}^{(m)} \right)^2.
\end{equation}
For an increasing chain of sets $\emptyset=A_0 \subseteq A_1 \subseteq A_2 \subseteq \dots \subseteq A_m=[m]$, \eqref{eq:AB} implies that
$$
Z_{A_0}^{(m)} \le Z_{A_1}^{(m)} \le Z_{A_2}^{(m)} \dots \le Z_{A_m}^{(m)}.
$$
For a given $0<\epsilon<0.1$, define $i_1$ as the largest integer between $0$ and $m$ such that 
$H_{A_{i_1}}^{(m)}< \epsilon$, and define $i_2$ as the smallest integer between $0$ and $m$ such that 
$H_{A_{i_2}}^{(m)}> 1- \epsilon$. According to Theorem~\ref{thm:m1}, $H_{A_i}^{(m)}< \epsilon$ for all $i\le i_1$ and $H_{A_i}^{(m)}>1- \epsilon$ for all $i\ge i_2$.
By Theorem~\ref{thm:m2}, we know that 
\begin{equation}\label{eq:nn}
i_2-i_1-1\le D(\epsilon).
\end{equation}
Since $H_{A_{i_1}}^{(m)}<\epsilon<0.1$, by \eqref{eq:os} we obtain that 
\begin{equation}\label{eq:ffp}
Z_{A_{i_1}}^{(m)}< 1/2.
\end{equation}
According to \eqref{eq:AB},
$$
\log_2 (Z_{A_i}^{(m)}) \le 2 \log_2 (Z_{A_{i+1}}^{(m)}),
$$
and so 
$$
\log_2 (Z_{A_i}^{(m)}) \le 2^j \log_2 (Z_{A_{i+j}}^{(m)}).
$$
For a given $0<\gamma<1$, define $i_3:= \lfloor i_1-\frac{1}{2} m^\gamma \rfloor$. 
If $i_3\ge 0$, then
$$
\log_2 (Z_{A_{i_3}}^{(m)}) \le 2^{m^\gamma/2} \log_2 (Z_{A_{i_1}}^{(m)})\le - 2^{m^\gamma/2}
\le -dm,
$$
where the second inequality follows from \eqref{eq:ffp} and the last inequality holds when $m$ is large enough. Therefore, for all $i\le i_3$,
$$
Z_{A_i}^{(m)}\le Z_{A_{i_3}}^{(m)} \le 2^{-dm} = n^{-d}=\delta_n.
$$
Thus we have
\begin{align*}
\{0,1,\dots,i_3\} \subseteq \left\{ i\in\{0,1,\dots,m\}:Z_{A_i}^{(m)} <\delta_n \right\}, \\
\left\{ i\in\{0,1,\dots,m\}:H_{A_i}^{(m)} > 1-\epsilon \right\} = \{i_2,i_2+1,\dots,m\}.
\end{align*}
Combining this with \eqref{eq:nn}, we obtain that
\begin{align*}
\left|\left\{ i\in\{0,1,\dots,m\}:H_{A_i}^{(m)} > 1-\epsilon \right\} \cup
\left\{ i\in\{0,1,\dots,m\}:Z_{A_i}^{(m)} < \delta_n \right\} \right| \\
\ge i_3+1 + m-i_2+1 \ge i_1-\frac{1}{2} m^\gamma + m-i_2 +1 \ge m-\frac{1}{2} m^\gamma-D(\epsilon)
\ge m-m^\gamma,
\end{align*}
where the last inequality holds when $m$ is large enough.

On the other hand, if $i_3<0$, then $i_1<\frac{1}{2} m^\gamma$, and by \eqref{eq:nn}, 
$i_2<\frac{1}{2} m^\gamma+D(\epsilon)+1$. Therefore
\begin{align*}
& \left|\left\{ i\in\{0,1,\dots,m\}:H_{A_i}^{(m)} > 1-\epsilon \right\} \cup
\left\{ i\in\{0,1,\dots,m\}:Z_{A_i}^{(m)} < \delta_n \right\} \right| \\
& \ge  m-i_2+1  \ge m-\frac{1}{2} m^\gamma-D(\epsilon)
\ge m-m^\gamma.
\end{align*}
This completes the proof of Theorem~\ref{thm:strong2}.

\subsection{Proof of Theorem~\ref{thm:m3}}
We first observe that there is a one-to-one mapping between increasing chains of sets and permutations on $[m]$.
Indeed, given $\pi\in S_m$, we can obtain an increasing chain of sets $\emptyset=A_0 \subseteq A_1 \subseteq A_2 \subseteq \dots \subseteq A_m=[m]$ by setting $A_i=\{\pi(1),\pi(2),\dots,\pi(i)\}$ for all $i\in[m]$.
On the other hand, given an increasing chain of sets $\emptyset=A_0 \subseteq A_1 \subseteq A_2 \subseteq \dots \subseteq A_m=[m]$, we can obtain a permutation $\pi\in S_m$ by setting 
$\pi(i)=A_i \setminus A_{i-1}$ for all $i\in[m]$.
Thus there are $m!$ distinct increasing chains of sets for a given $m$.
Let us fix $m$ and list all the $m!$ distinct increasing chains of sets as follows:
\begin{align*}
\emptyset=A_0(1) \subseteq A_1(1) \subseteq A_2(1) \subseteq & \dots \subseteq A_m(1)=[m], \\
\emptyset=A_0(2) \subseteq A_1(2) \subseteq A_2(2) \subseteq & \dots \subseteq A_m(2)=[m], \\
\emptyset=A_0(3) \subseteq A_1(3) \subseteq A_2(3) \subseteq & \dots \subseteq A_m(3)=[m], \\
\vdots \hspace*{0.2in} \vdots \hspace*{0.2in} \vdots \hspace*{0.2in} & \vdots \\
\emptyset=A_0(m!) \subseteq A_1(m!) \subseteq A_2(m!) \subseteq & \dots \subseteq A_m(m!)=[m].
\end{align*}
Notice that for every $i\in\{0,1,2,\dots,m\}$, $|A_i(1)|=|A_i(2)|=|A_i(3)|=\dots=|A_i(m!)|=i$.
There are $\frac{m!}{i!(m-i)!}$ subsets of $[m]$ with cardinality $i$. By symmetry, each of them appears the same number of times in $(A_i(1),A_i(2),A_i(3),\dots,A_i(m!))$. Thus each subset with cardinality $i$ appears $i!(m-i)!$ times in $(A_i(1),A_i(2),A_i(3),\dots,A_i(m!))$. In other words, each subset $A\subseteq[m]$ appears
$|A|!(m-|A|)!$ times in $(A_i(j):i\in\{0,1,2,\dots,m\},j\in[m!])$.

For any $0<\epsilon<0.1$, define 
$$
\cS(\epsilon) := \left\{(i,j):i\in\{0,1,2,\dots,m\},j\in[m!], H_{A_i(j)}^{(m)} \le 1-\epsilon,
 Z_{A_i(j)}^{(m)} \ge \delta_n \right\}.
$$
Then by Theorem~\ref{thm:strong2}, we know that for any $0<\gamma<1/2$ and any given $j\in[m!]$,
$$
\left|\left\{ i\in\{0,1,\dots,m\}: H_{A_i(j)}^{(m)} \le 1-\epsilon,
 Z_{A_i(j)}^{(m)} \ge \delta_n 
\right\} \right| \le m^\gamma
\text{~~~for all~} m> M(\epsilon,\gamma).
$$
Consequently, for all $m> M(\epsilon,\gamma)$,
\begin{equation}\label{eq:sw}
|\cS(\epsilon)| \le (m!) m^\gamma.
\end{equation}
We further define
$$
\cA(\epsilon) := \left\{A\subseteq[m]: H_A^{(m)} \le 1-\epsilon,
Z_A^{(m)} \ge \delta_n \right\}.
$$
By the arguments above, we have
$$
|\cS(\epsilon)| = \sum_{A\subseteq \cA(\epsilon)} |A|!(m-|A|)!.
$$
It is easy to see that $i!(m-i)!\ge \lfloor m/2 \rfloor!(m-\lfloor m/2 \rfloor)!$ for all $i\in\{0,1,2,\dots,m\}$. Therefore
$$
|\cS(\epsilon)| \ge \lfloor m/2 \rfloor!(m-\lfloor m/2 \rfloor)! |\cA(\epsilon)|.
$$
Combining this with \eqref{eq:sw}, we obtain that
$$
|\cA(\epsilon)| \le \binom{m}{\lfloor m/2 \rfloor} m^\gamma.
$$
Consequently,
$$
\frac{|\cA(\epsilon)|}{2^m} \le m^\gamma \frac{\binom{m}{\lfloor m/2 \rfloor}}{2^m} .
$$
By Stirling's formula,
\begin{equation}\label{eq:stl}
\frac{\binom{m}{\lfloor m/2 \rfloor}}{2^m} = \sqrt{\frac{2}{\pi m}} (1+o_m(1)) .
\end{equation}
Since $\sqrt{2/\pi}<1$, we conclude that for all $m>M(\epsilon,\gamma)$,
$$
\frac{|\cA(\epsilon)|}{2^m}
 \le m^{\gamma-1/2}.
$$
This completes the proof of Theorem~\ref{thm:m3}.

\subsection{Proof of Theorem~\ref{thm:cac}}
We first show that the code rate of $\cT(m,\delta_n)$ approaches the channel capacity $I(W)$, i.e., 
$$
|\cG(m,\delta_n)| \ge 2^m (I(W)-o(1)).
$$
By \eqref{eq:kw}, for all $0<\epsilon<1$, we have
$$
(1-\epsilon) \left| \left\{A\subseteq[m]:  H_A^{(m)} > 1-\epsilon \right\} \right| <
\sum_{A\subseteq[m]} H_A^{(m)} = 2^m (1-I(W)).
$$
Therefore,
$$
\left| \left\{A\subseteq[m]:  H_A^{(m)} > 1-\epsilon \right\} \right| < \frac{1}{1-\epsilon}2^m (1-I(W)).
$$
According to Theorem~\ref{thm:m3}, for $0<\epsilon<0.1$, $0<\gamma<1/2$ and $m>M(\epsilon,\gamma)$,
\begin{align*}
\left| \cG(m,\delta_n) \right|  \ge &
\left| \left\{A\subseteq[m]:  H_A^{(m)} > 1-\epsilon \right\}
\cup \left\{A\subseteq[m]:  Z_A^{(m)} < \delta_n \right\}\right| \\
& - \left| \left\{A\subseteq[m]:  H_A^{(m)} > 1-\epsilon \right\}  \right| \\
\ge & 2^m(1-m^{\gamma-1/2}) - \frac{1}{1-\epsilon}2^m (1-I(W)).
\end{align*}
The last line can be made arbitrarily close to $2^m I(W)$ if we set $\epsilon$ to be small enough and $m$ to be large enough. Thus the code rate of $\cT(m,\delta_n)$ approaches $I(W)$.

Next we prove that the decoding error of $\cT(m,\delta_n)$ goes to $0$ under the successive decoder that is similar to the one used for polar codes, i.e., we decode $U_A^{(m)}$ one by one using the channel outputs $Y^{(m)}$ and the previously decoded inputs $U_{<A}^{(m)}$.
The decoding order is from small to large sets according to the order defined in Section~\ref{sect:intro}, i.e., we decode $U_A$ before decoding $U_B$ if $A<B$.
According to \eqref{eq:pez},
for every $A\in \cG(m,\delta_n)$, the error probability of decoding $U_A^{(m)}$ from $Y^{(m)}$ and $U_{<A}^{(m)}$ is at most
$$
P_e(U_A^{(m)}|Y^{(m)},U_{<A}^{(m)})
\le Z(U_A^{(m)}|Y^{(m)},U_{<A}^{(m)}) = Z_A^{(m)} <n^{-2}.
$$
By the union bound, the error probability of decoding the whole codeword under successive decoder is at most $n^{-1}$. Thus we conclude that the code $\cT(m,\delta_n)$ achieves the capacity of $W$.

\subsection{Proof of Theorem \ref{order2}}
Let $A \prec B$. Define $A':= \{a_1,\dots,a_{|B|}\}$ and note that by assumption, $A'$ is pointwise smaller than $B$.

We first apply \eqref{eq:jw} repeatedly to obtain    
\begin{align}
H_A^{(m)} \ge H_{\{a_1,\dots,a_{|A|-1}\}}^{(m)} \ge H_{\{a_1,\dots,a_{|A|-2}\}}^{(m)} \ge \dots \ge H_{A'}^{(m)}.
\end{align}
We then apply Lemma \ref{lm:is} repeatedly to obtain
\begin{align}
H_{A'}^{(m)} = H_{\{a_1,\dots,a_{|B|-1},a_{|B|}\}}^{(m)} \ge H_{\{a_1,\dots,a_{|B|-1},b_{|B|}\}}^{(m)} \ge H_{\{a_1,\dots,b_{|B|-1},b_{|B|}\}}^{(m)} \ge \dots \ge H_{B}^{(m)}.
\end{align}
Therefore $H_A^{(m)} \ge H_{A'}^{(m)} \ge H_{B}^{(m)}$.

\bibliographystyle{IEEEtran}
\bibliography{RM}

\begin{thebibliography}{10}
\providecommand{\url}[1]{#1}
\csname url@samestyle\endcsname
\providecommand{\newblock}{\relax}
\providecommand{\bibinfo}[2]{#2}
\providecommand{\BIBentrySTDinterwordspacing}{\spaceskip=0pt\relax}
\providecommand{\BIBentryALTinterwordstretchfactor}{4}
\providecommand{\BIBentryALTinterwordspacing}{\spaceskip=\fontdimen2\font plus
\BIBentryALTinterwordstretchfactor\fontdimen3\font minus
  \fontdimen4\font\relax}
\providecommand{\BIBforeignlanguage}[2]{{%
\expandafter\ifx\csname l@#1\endcsname\relax
\typeout{** WARNING: IEEEtran.bst: No hyphenation pattern has been}%
\typeout{** loaded for the language `#1'. Using the pattern for}%
\typeout{** the default language instead.}%
\else
\language=\csname l@#1\endcsname
\fi
#2}}
\providecommand{\BIBdecl}{\relax}
\BIBdecl

\bibitem{Kudekar17}
S.~Kudekar, S.~Kumar, M.~Mondelli, H.~D. Pfister, E.~{\c{S}}a{\c{s}}oǧlu, and
  R.~Urbanke, ``Reed--{M}uller codes achieve capacity on erasure channels,''
  \emph{IEEE Transactions on Information Theory}, vol.~63, no.~7, pp.
  4298--4316, 2017.

\bibitem{Kudekar16STOC}
S.~Kudekar, S.~Kumar, M.~Mondelli, H.~D. Pfister, E.~{\c{S}}a{\c{s}}o{\u{g}}lu,
  and R.~Urbanke, ``Reed-{M}uller codes achieve capacity on erasure channels,''
  in \emph{Proceedings of the forty-eighth annual ACM symposium on Theory of
  Computing}.\hskip 1em plus 0.5em minus 0.4em\relax ACM, 2016, pp. 658--669.

\bibitem{Abbe15}
E.~Abbe, A.~Shpilka, and A.~Wigderson, ``Reed--{M}uller codes for random
  erasures and errors,'' \emph{IEEE Transactions on Information Theory},
  vol.~61, no.~10, pp. 5229--5252, 2015.

\bibitem{Sberlo18}
O.~Sberlo and A.~Shpilka, ``On the performance of {R}eed-{M}uller codes with
  respect to random errors and erasures,'' \emph{arXiv:1811.12447}, 2018.

\bibitem{Arikan09}
E.~Ar{\i}kan, ``Channel polarization: {A} method for constructing
  capacity-achieving codes for symmetric binary-input memoryless channels,''
  \emph{IEEE Transactions on Information Theory}, vol.~55, no.~7, pp.
  3051--3073, 2009.

\bibitem{ArikanTelatar}
E.~Ar{\i}kan and E.~Telatar, ``On the rate of channel polarization,'' in
  \emph{2009 IEEE International Symposium on Information Theory}.\hskip 1em
  plus 0.5em minus 0.4em\relax IEEE, 2009, pp. 1493--1495.

\bibitem{Tal13}
I.~Tal and A.~Vardy, ``How to construct polar codes,'' \emph{IEEE Transactions
  on Information Theory}, vol.~59, no.~10, pp. 6562--6582, 2013.

\bibitem{Hassani09}
S.~H. Hassani, S.~B. Korada, and R.~Urbanke, ``The compound capacity of polar
  codes,'' in \emph{47th Annual Allerton Conference on Communication, Control,
  and Computing}.\hskip 1em plus 0.5em minus 0.4em\relax IEEE, 2009, pp.
  16--21.

\bibitem{Hassani14}
S.~H. Hassani, K.~Alishahi, and R.~Urbanke, ``Finite-length scaling for polar
  codes,'' \emph{IEEE Transactions on Information Theory}, vol.~60, no.~10, pp.
  5875--5898, 2014.

\bibitem{Hassani18}
H.~Hassani, S.~Kudekar, O.~Ordentlich, Y.~Polyanskiy, and R.~Urbanke, ``Almost
  optimal scaling of {R}eed-{M}uller codes on {BEC} and {BSC} channels,'' in
  \emph{2018 IEEE International Symposium on Information Theory (ISIT)}, June
  2018, pp. 311--315.

\bibitem{Mondelli14}
M.~Mondelli, S.~H. Hassani, and R.~L. Urbanke, ``From polar to {R}eed-{M}uller
  codes: {A} technique to improve the finite-length performance,'' \emph{IEEE
  Transactions on Communications}, vol.~62, no.~9, pp. 3084--3091, 2014.

\bibitem{Guruswami15}
V.~Guruswami and P.~Xia, ``Polar codes: {S}peed of polarization and polynomial
  gap to capacity,'' \emph{IEEE Transactions on Information Theory}, vol.~61,
  no.~1, pp. 3--16, 2015.

\bibitem{Blasiok18}
J.~B{\l}asiok, V.~Guruswami, P.~Nakkiran, A.~Rudra, and M.~Sudan, ``General
  strong polarization,'' in \emph{Proceedings of the 50th Annual ACM SIGACT
  Symposium on Theory of Computing}.\hskip 1em plus 0.5em minus 0.4em\relax
  ACM, 2018, pp. 485--492.

\bibitem{Tal15}
I.~Tal and A.~Vardy, ``List decoding of polar codes,'' \emph{IEEE Transactions
  on Information Theory}, vol.~61, no.~5, pp. 2213--2226, 2015.

\bibitem{3gpp}
``Final report of 3{GPP TSG RAN WG}1 \#87 v1.0.0,''
  http://www.3gpp.org/ftp/tsg\_ran/WG1\_RL1/TSGR1\_87/Report/.

\bibitem{Reed54}
I.~Reed, ``A class of multiple-error-correcting codes and the decoding
  scheme,'' \emph{Transactions of the IRE Professional Group on Information
  Theory}, vol.~4, no.~4, pp. 38--49, 1954.

\bibitem{Dumer04}
I.~Dumer, ``Recursive decoding and its performance for low-rate {R}eed-{M}uller
  codes,'' \emph{IEEE Transactions on Information Theory}, vol.~50, no.~5, pp.
  811--823, 2004.

\bibitem{Saptharishi17}
R.~Saptharishi, A.~Shpilka, and B.~L. Volk, ``Efficiently decoding
  {R}eed--{M}uller codes from random errors,'' \emph{IEEE Transactions on
  Information Theory}, vol.~63, no.~4, pp. 1954--1960, 2017.

\bibitem{Santi18}
E.~Santi, C.~H{\"a}ger, and H.~D. Pfister, ``Decoding {R}eed-{M}uller codes
  using minimum-weight parity checks,'' 2018, arXiv:1804.10319.

\bibitem{YA18}
M.~Ye and E.~Abbe, ``Recursive projection-aggregation decoding of
  {R}eed-{M}uller codes,'' 2018, to appear on arXiv.

\bibitem{Kahn88}
J.~Kahn, G.~Kalai, and N.~Linial, ``The influence of variables on boolean
  functions,'' in \emph{Proceedings of the 29th Annual Symposium on Foundations
  of Computer Science}.\hskip 1em plus 0.5em minus 0.4em\relax IEEE Computer
  Society, 1988, pp. 68--80.

\bibitem{Talagrand94}
M.~Talagrand, ``On {R}usso's approximate zero-one law,'' \emph{The Annals of
  Probability}, vol.~22, no.~3, pp. 1576--1587, 1994.

\bibitem{Friedgut96}
E.~Friedgut and G.~Kalai, ``Every monotone graph property has a sharp
  threshold,'' \emph{Proceedings of the American mathematical Society}, vol.
  124, no.~10, pp. 2993--3002, 1996.

\bibitem{Bourgain92}
J.~Bourgain, J.~Kahn, G.~Kalai, Y.~Katznelson, and N.~Linial, ``The influence
  of variables in product spaces,'' \emph{Israel Journal of Mathematics},
  vol.~77, no. 1-2, pp. 55--64, 1992.

\bibitem{Wigderson16}
Y.~Wigderson, ``Algebraic properties of tensor product matrices, with
  applications to coding,'' \emph{Senior thesis, Princeton University}, 2016.

\bibitem{Yuval15}
E.~Abbe and Y.~Wigderson, ``High-girth matrices and polarization,'' in
  \emph{2015 IEEE International Symposium on Information Theory (ISIT)}.\hskip
  1em plus 0.5em minus 0.4em\relax IEEE, 2015, pp. 2461--2465.

\bibitem{Macwilliams77}
F.~J. MacWilliams and N.~J.~A. Sloane, \emph{The theory of error-correcting
  codes}.\hskip 1em plus 0.5em minus 0.4em\relax Elsevier, 1977.

\bibitem{Korada09}
S.~B. Korada, ``Polar codes for channel and source coding,'' Ph.D.
  dissertation, {\'E}COLE POLYTECHNIQUE F{\'E}D{\'E}RALE DE LAUSANNE, 2009.

\bibitem{Weiss62}
E.~Weiss, ``Compression and coding,'' \emph{IRE Transactions on Information
  Theory}, vol.~8, no.~3, pp. 256--257, 1962.

\bibitem{Allard72}
P.~E. Allard and A.~W. Bridgewater, ``A source encoding technique using
  algebraic codes,'' in \emph{Proc. 1972 Canadian computer conference}, 1972,
  pp. 201--213.

\end{thebibliography}

\section{Additional results}

In this section, we present three additional results. First, we provide a sufficient condition for RM codes to achieve capacity (called the gap property); second, we show the equivalence between source and channel coding using RM codes; third, we show that the twin code is indeed the same as the RM codes up to $n=16$ for the BSC.
In order to state and prove the first result, we need some more notation.
Note that according to our total order on the subsets $A\subseteq[m]$, for all $A$ with cardinality $|A|=i$, we have
$$
\{1,2,\dots,i\} \le A \le \{m-i+1,m-i+2,\dots,m\}. 
$$
The following corollary follows from Theorem~\ref{order2}.
\begin{corollary}\label{cr:od}
For every BMS channel $W$ and every $A\subseteq[m]$ with cardinality $|A|=i$,
\begin{align}
H_{\{1,2,\dots,i\}}^{(m)} \ge H_A^{(m)} \ge H_{\{m-i+1,m-i+2,\dots,m\}}^{(m)}. \label{eq:too} \\
Z_{\{1,2,\dots,i\}}^{(m)} \ge Z_A^{(m)} \ge Z_{\{m-i+1,m-i+2,\dots,m\}}^{(m)}. \label{eq:kml}
\end{align}
\end{corollary}

For a fixed $m$ and $i\in[m]$, define
\begin{align}
H_{i,\max}^{(m)} & := H_{\{1,2,\dots,i\}}^{(m)}, \nonumber \\
H_{i,\min}^{(m)} & := H_{\{m-i+1,m-i+2,\dots,m\}}^{(m)}, \nonumber \\
H_{i,\avg}^{(m)} & := \frac{1}{\binom{m}{i}} \sum_{A\subseteq[m],|A|=i} H_A^{(m)}.
\label{eq:oo}
\end{align}
We further define 
$$
H_{0,\max}^{(m)} = H_{0,\min}^{(m)} = H_{0,\avg}^{(m)}  = H_{\emptyset}^{(m)}.
$$
By \eqref{eq:too}, for all $i\in\{0,1,2,\dots,m\}$,
$$
H_{i,\max}^{(m)} \ge H_{i,\avg}^{(m)} \ge H_{i,\min}^{(m)}.
$$

Applying Theorem~\ref{thm:m1} to the increasing chains of sets
$$
\emptyset \subseteq \{1\} \subseteq \{1,2\} \subseteq \{1,2,3\} \subseteq \dots\subseteq \{1,2,3,\dots,m\}
$$
and 
$$
\emptyset \subseteq \{m\} \subseteq \{m-1,m\} \subseteq \{m-2,m-1,m\} \subseteq \dots \subseteq
\{1,2,3,\dots,m\}
$$
respectively, we obtain the following corollary.
\begin{corollary}\label{cr:dh1}
For every BMS channel $W$ and every $m>0$, both sequences $\{H_{i,\max}^{(m)}\}_{i=0}^m$ and $\{H_{i,\min}^{(m)}\}_{i=0}^m$  increase with $i$, i.e.,
\begin{align}
H_{0,\max}^{(m)} \le H_{1,\max}^{(m)} \le H_{2,\max}^{(m)} \le \dots \le H_{m,\max}^{(m)}, \label{eq:mr3} \\
H_{0,\min}^{(m)} \le H_{1,\min}^{(m)} \le H_{2,\min}^{(m)} \le \dots \le H_{m,\min}^{(m)}, \label{eq:mr4} 
\end{align}
\end{corollary}

This corollary together with Corollary~\ref{cr:od} immediately implies the following claim.
\begin{corollary} \label{cr:dh2}
For every BMS channel $W$ and every $m>0$,
\begin{align}
H_{i,\max}^{(m)} \ge H_A^{(m)}, \text{~for all~} A\subseteq[m], |A|\le i; \label{eq:mr1} \\
H_{i,\min}^{(m)} \le H_A^{(m)}, \text{~for all~} A\subseteq[m], |A|\ge i. \label{eq:mr2}
\end{align}
\end{corollary}

Finally, the following corollary follows immediately from Theorem~\ref{thm:m2}:
\begin{corollary}\label{cr:dh3}
For every BMS channel $W$ and every $\epsilon>0$, there is a constant $D(\epsilon)$ (which is independent of $m$ and $W$) such that for every positive integer $m$,
\begin{align*}
\left|\left\{ i\in\{0,1,\dots,m\}:\epsilon < H_{i,\max}^{(m)} < 1-\epsilon \right\} \right| < D(\epsilon), \\
\left|\left\{ i\in\{0,1,\dots,m\}:\epsilon < H_{i,\min}^{(m)} < 1-\epsilon \right\} \right| < D(\epsilon).
\end{align*}
\end{corollary}

In Section~\ref{sect:proofs}, we also prove the following corollary of Theorem~\ref{thm:m1} that states that for every BMS channel $W$ and every $m>0$, the sequence $\{H_{i,\avg}^{(m)}\}_{i=0}^m$ increases with $i$:
\begin{corollary}\label{cr:avg}
$$
H_{0,\avg}^{(m)} \le H_{1,\avg}^{(m)} \le H_{2,\avg}^{(m)} \le \dots \le H_{m,\avg}^{(m)}. 
$$
\end{corollary}

\subsection{A sufficient condition for RM codes to achieve capacity}
By Corollary~\ref{cr:dh1} and \ref{cr:dh3} we know that both sequences 
$\{H_{i,\max}^{(m)}\}_{i=0}^m$ and $\{H_{i,\min}^{(m)}\}_{i=0}^m$ are increasing sequences with sharp transitions from $0$ to $1$.
Let $\theta_{\max}^{(m)}$ and $\theta_{\min}^{(m)}$ be the locations of transition for these two sequences, respectively.
More precisely, let
\begin{align*}
\theta_{\max}^{(m)} := \max \left( \left\{ i\in\{0,1,\dots,m\}:H_{i,\max}^{(m)}\le 1/2 \right\} \right), \\
\theta_{\min}^{(m)} := \max \left( \left\{ i\in\{0,1,\dots,m\}:H_{i,\min}^{(m)}\le 1/2 \right\} \right),
\end{align*}
i.e., $\theta_{\max}^{(m)}$ is the largest index $i$ such that $H_{i,\max}^{(m)}\le 1/2$ and $\theta_{\min}^{(m)}$ is the largest index $i$ such that $H_{i,\min}^{(m)}\le 1/2$.
Clearly, we always have $\theta_{\max}^{(m)} < \theta_{\min}^{(m)}$.
Figure~\ref{fig:ilt} below illustrates the sharp transitions of $\{H_{i,\max}^{(m)}\}_{i=0}^m$ and $\{H_{i,\min}^{(m)}\}_{i=0}^m$ and the definition of $\theta_{\max}^{(m)}$ and $\theta_{\min}^{(m)}$.

Our next theorem shows that if $\theta_{\max}^{(m)}$ and $\theta_{\min}^{(m)}$ are close enough for large $m$, then RM codes achieve capacity of $W$. 
\begin{theorem}\label{thm:m4}
Let $W$ be a BMS channel.
If
\begin{equation}\label{eq:qqq}
\theta_{\min}^{(m)}-\theta_{\max}^{(m)}=o(\sqrt{m}),
\end{equation}
then Reed-Muller codes achieve capacity of $W$.
\end{theorem}

A consequence of Theorem~\ref{thm:m4} is recorded in the following corollary.

\begin{corollary}\label{cr:c4}
If there exist a constant $0\le\gamma<1$ and a nonnegative integer-valued function $\Delta:\mathbb{Z}\to\mathbb{Z}$ such that 
$$
\Delta(m)=o(\sqrt{m})
$$ 
and
\begin{equation}\label{eq:jj}
H_{i,\max}^{(m)} \le H_{i+\Delta(m),\min}^{(m)} + \gamma
\text{~~for all~} 0\le i\le m-\Delta(m) \text{~and all~} m>0,
\end{equation}
then Reed-Muller codes achieve capacity of $W$.
\end{corollary}
Taking $\gamma=0$ and $\Delta(m)=1$ for all $m>0$ in the corollary above, we immediately obtain the following result.
\begin{corollary}[Gap property]\label{cr:sv}
If
\begin{equation}\label{eq:nmv}
H_{i,\max}^{(m)} \le H_{i+1,\min}^{(m)} 
\text{~~for all~} 0\le i\le m-1 \text{~and all~} m>0,
\end{equation}
then Reed-Muller codes achieve capacity of $W$.
\end{corollary}
We have numerically verified \eqref{eq:nmv} on BEC for code length up to $256$; see Fig.~\ref{fig:ant} and Section~\ref{sect:sim} for a detailed discussion. We also proved \eqref{eq:nmv} up to dimension 16 for the BSC.

\begin{figure}[h] 
\includegraphics[width=\textwidth]{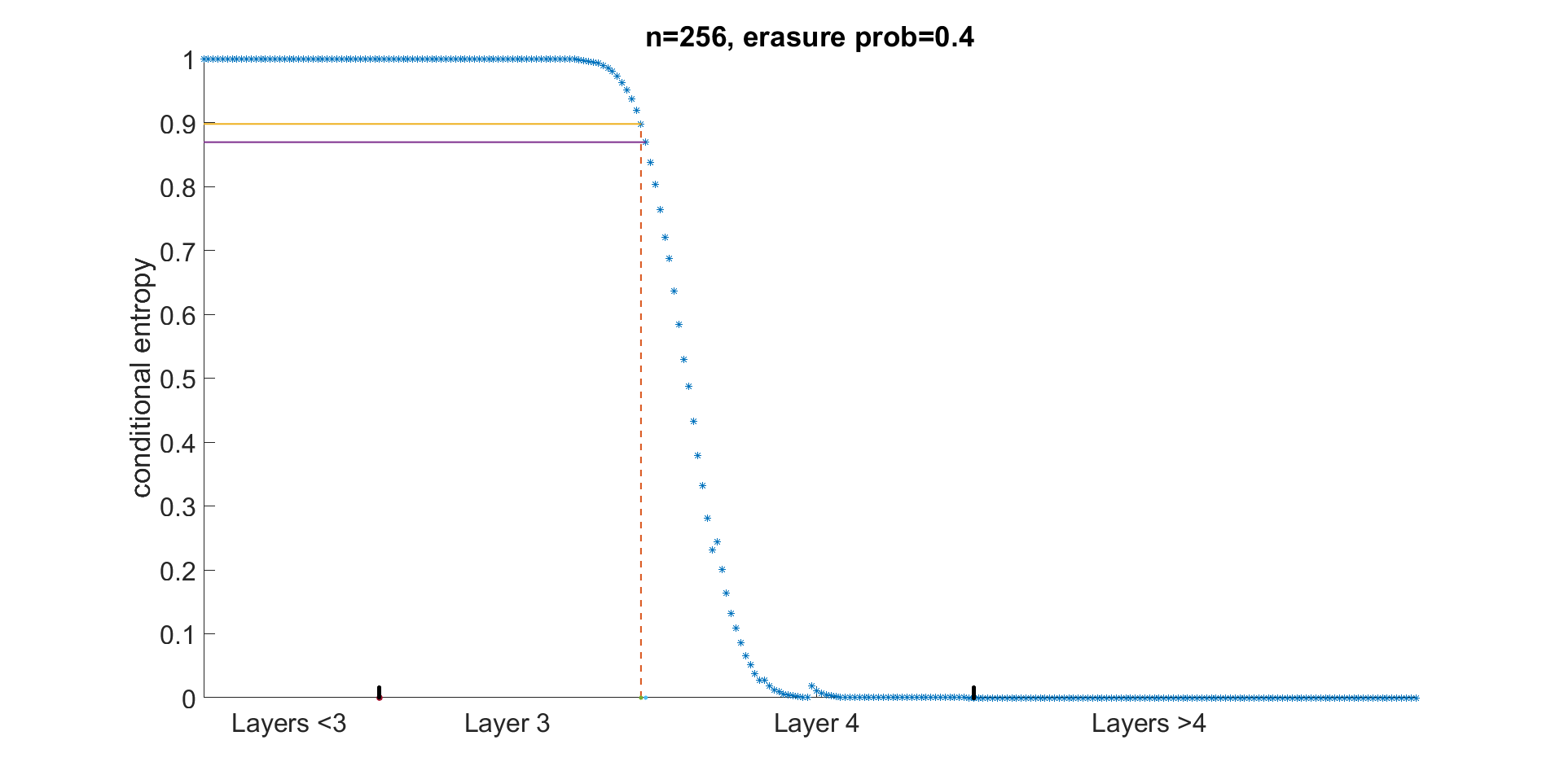}
\caption{Simulation results for Reed-Muller codes with length $n=256$ over a BEC, where the channel erasure probability is set to be $0.4$. This figure illustrates the gap property between Layer 3 and Layer 4.}
\label{fig:ant}
\end{figure}

\begin{figure}[h] 
{\centering \textbf{Sharp transitions of $\{H_{i,\max}^{(m)}\}_{i=0}^m$ and $\{H_{i,\min}^{(m)}\}_{i=0}^m$}\par\medskip}
\includegraphics[width=\textwidth]{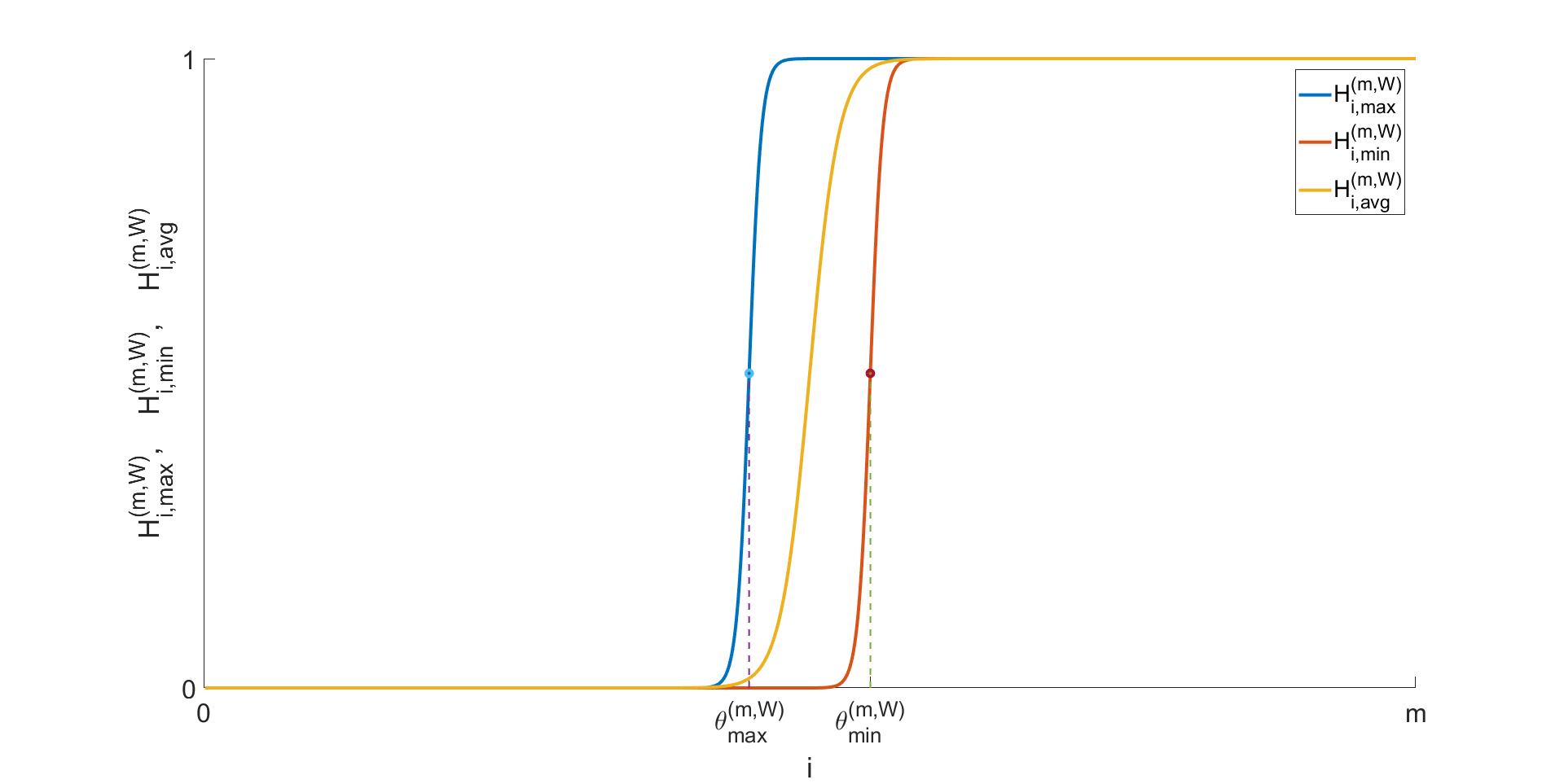}
\caption{All three sequences $\{H_{i,\max}^{(m)}\}_{i=0}^m,\{H_{i,\min}^{(m)}\}_{i=0}^m$ and $\{H_{i,\avg}^{(m)}\}_{i=0}^m$ increase with $i$. Both sequences 
$\{H_{i,\max}^{(m)}\}_{i=0}^m$ and $\{H_{i,\min}^{(m)}\}_{i=0}^m$ have sharp transition from $0$ to $1$, and the locations of transition are denoted as $\theta_{\max}^{(m)}$ and $\theta_{\min}^{(m)}$ respectively. If $\theta_{\min}^{(m)}-\theta_{\max}^{(m)}=o(\sqrt{m})$, then Reed-Muller codes achieve capacity of $W$.} 
\label{fig:ilt}
\end{figure}

\subsection{Equivalence between source and channel coding using RM codes}
We use BSC$(p)$ to denote Binary Symmetric Channel with crossover probability $p$. The capacity of BSC$(p)$ is $1-h(p)$, where $h(p):=-p\log_2 p-(1-p)\log_2 (1-p)$ is the binary entropy function. It is well known that the optimal rate to losslessly compress i.i.d. Bernoulli-$p$ source is $h(p)$.
In this paper we show that for RM codes, the channel coding problem for BSC is equivalent to the source coding problem:
\begin{theorem}\label{thm:m5}
Reed-Muller codes achieve capacity of BSC$(p)$ if and only if they can be used to losslessly compress i.i.d. Bernoulli-$p$ random variables with optimal rate $h(p)$.
\end{theorem}
The proof of this theorem uses arguments similar to \cite{Weiss62,Allard72}; see Section~\ref{sect:p5}.

\subsection{Twin code is the same as RM code up to $n=16$ for BSC} \label{sect:tisr}
We show that the twin code is the same as the RM code up to $n=16$ for BSC. Our claim follows immediately from the following proposition:
\begin{proposition} \label{prop:jww}
For BSC channels and $m\le 4$, if two subsets $A,B\subseteq[m]$ satisfy that $|A|>|B|$, then $H_A^{(m)}\ge H_B^{(m)}$.
\end{proposition}
The proof is given in Section~\ref{sect:p6}.

\section{Proofs of the additional results}\label{sect:proofs}

\subsection{Proof of Corollary~\ref{cr:avg}}
Let us fix an $i\in\{0,1,\dots,m-1\}$.
We consider all the distinct pair of subsets $A,B\subseteq[m]$ such that $A\subseteq B, |A|=i, |B|=i+1$.
Obviously for each fixed subset $A$, there are $m-i$ possible choices of $B$, and for each fixed subset $B$, there are $i+1$ possible choices of $A$.
By Theorem~\ref{thm:m1} we know that 
\begin{equation}\label{eq:lq}
H_A^{(m)} \le H_B^{(m)}
\end{equation}
for all $A\subseteq B$.
Summing \eqref{eq:lq} over all distinct pairs of subsets $A,B\subseteq[m]$ such that $A\subseteq B, |A|=i, |B|=i+1$, we obtain that
$$
(m-i)\sum_{A\subseteq[m],|A|=i} H_A^{(m)} \le (i+1) \sum_{B\subseteq[m],|B|=i+1} H_B^{(m)}.
$$
According to the definition \eqref{eq:oo} of $H_{i,\avg}^{(m)}$, this is equivalent to
$$
H_{i,\avg}^{(m)} \le H_{i+1,\avg}^{(m)}.
$$
This completes the proof of Corollary~\ref{cr:avg}.

\subsection{Proof of Theorem~\ref{thm:m4} and Corollary~\ref{cr:c4}}

\vspace*{0.1in}\underline{\em Proof of Theorem~\ref{thm:m4}:} 
Take $\epsilon=0.1$ and the corresponding constant $D(0.1)$ in Corollary~\ref{cr:dh3}.
Let
\begin{equation}\label{eq:ab}
\beta:= \lfloor \theta_{\max}^{(m)} - D(0.1) \rfloor, \quad \quad
\alpha := \lfloor \beta - m^{1/3} \rfloor  .
\end{equation}
We claim that the family of Reed-Muller codes $\cR(m,\alpha)$ achieves the capacity of $W$ when $m$ goes to infinity. Notice that by definition \eqref{eq:ab}, $\alpha$ is a function of $m$ and $W$.
To prove this claim we only need to show two facts: 
First, the decoding error of $\cR(m,\alpha)$ goes to $0$.
Second, the rate of $\cR(m,\alpha)$ approaches $I(W)$, i.e.,
\begin{equation}\label{eq:rcd}
 \sum_{i=0}^\alpha \binom{m}{i} \ge n (I(W) - o(1)).
\end{equation}

We start by showing that the decoding error goes to $0$.
According to Corollary~\ref{cr:dh1} and Corollary~\ref{cr:dh3},
$$
H_{\{1,2,\dots,\beta\}}^{(m)} =H_{\beta,\max}^{(m)}\le 0.1
$$
Then by \eqref{eq:os} we obtain that 
$$
Z_{\{1,2,\dots,\beta\}}^{(m)} < 1/2.
$$
By \eqref{eq:AB},
$$
\log_2 \left(Z_{\{1,2,\dots,\alpha\}}^{(m)}\right)
 \le 2^{\beta-\alpha} \log_2 \left( Z_{\{1,2,\dots,\beta\}}^{(m)} \right)
\le 2^{m^{1/3}} \log_2 \left( Z_{\{1,2,\dots,\beta\}}^{(m)} \right)
< -2^{m^{1/3}}.
$$
Since $2^{m^{1/3}} \gg 2m$, we have
$$
Z_{\{1,2,\dots,\alpha\}}^{(m)} \le 2^{-2m} = n^{-2}.
$$
Recall that the code length $n=2^m$. 
Note that \eqref{eq:AB} also implies that $Z_{\{1,2,\dots,i\}}^{(m)} \le Z_{\{1,2,\dots,i+1\}}^{(m)}$ for all $0\le i<m$. Therefore, for all $i\le \alpha$,
$$
Z_{\{1,2,\dots,i\}}^{(m)} \le Z_{\{1,2,\dots,\alpha\}}^{(m)} \le  n^{-2}.
$$
Combining this with \eqref{eq:kml}, we conclude that for all $A\subseteq [m]$ with cardinality $|A|\le \alpha$,
$$
Z_A^{(m)} \le Z_{\{1,2,\dots,|A|\}}^{(m)} \le  n^{-2}.
$$
We again use the successive decoder.
By the union bound, the decoding error of the whole codeword is upper bounded by the sum of decoding errors of each individual input:
\begin{align*}
P_e(\cR(m,\alpha)) & \le \sum_{A\subseteq[m], |A| \le \alpha} P_e \Big(U_A^{(m)} \Big|Y^{(m)},U_{<A}^{(m)} \Big) \\
& \le \sum_{A\subseteq[m], |A| \le \alpha} Z \Big(U_A^{(m)} \Big|Y^{(m)},U_{<A}^{(m)} \Big) \\
& \le \sum_{A\subseteq[m], |A| \le \alpha} n^{-2} \le n^{-1},
\end{align*}
where the second inequality follows from \eqref{eq:pez}.
Therefore the decoding error of $\cR(m,\alpha)$ does go to $0$.

Next we prove \eqref{eq:rcd}.
Let
$$
t(\epsilon):= \lceil \theta_{\min}^{(m)} + D(\epsilon) \rceil.
$$
According to Corollary~\ref{cr:dh1} and Corollary~\ref{cr:dh3},
$$
H_{t(\epsilon),\min}^{(m)} \ge 1-\epsilon.
$$
By Corollary~\ref{cr:dh2}, for all $A\subseteq[m]$ with cardinality $|A|\ge t(\epsilon)$,
$$
H_A^{(m)} \ge H_{t(\epsilon),\min}^{(m)} \ge 1-\epsilon.
$$
By \eqref{eq:kw}, we have
$$
(1-\epsilon) \sum_{i=t(\epsilon)}^m \binom{m}{i}   \le
\sum_{A\subseteq[m], |A|\ge t(\epsilon)} H_A^{(m)} \le 
\sum_{A\subseteq[m]} H_A^{(m)} 
= n(1-I(W)).
$$
Therefore,
\begin{equation} \label{eq:hhw}
\sum_{i=t(\epsilon)}^m \binom{m}{i}   \le  \frac{1}{1-\epsilon} n(1-I(W)).
\end{equation}
Since $\binom{m}{i} \le \binom{m}{\lfloor m/2 \rfloor}$ for all $i$ and by \eqref{eq:stl}, 
$\binom{m}{\lfloor m/2 \rfloor} < \frac{n}{\sqrt{m}}$ for large $m$, we have
\begin{align*}
\sum_{i=\alpha+1}^{t(\epsilon)-1} \binom{m}{i}
& \le (t(\epsilon)-\alpha-1) \frac{n}{\sqrt{m}}
\le (\theta_{\min}^{(m)} + D(\epsilon) - (\theta_{\max}^{(m)} - D(0.1)-m^{1/3})) 
\frac{n}{\sqrt{m}}  \\
& = ( o(\sqrt{m}) + D(\epsilon) + D(0.1) + m^{1/3} )
\frac{n}{\sqrt{m}} =o(n).
\end{align*}
Combining this with \eqref{eq:hhw}, we obtain that 
$$
\sum_{i=0}^\alpha \binom{m}{i} \ge n- \frac{1}{1-\epsilon} n(1-I(W)) -o(n).
$$
The right-hand side can be made arbitrarily close to $nI(W)$ as long as we set $\epsilon$ to be small enough and $m$ to be large enough. This completes the proof of \eqref{eq:rcd} and establishes Theorem~\ref{thm:m4}. 
\hfill \qedsymbol

\vspace*{0.1in}\underline{\em Proof of Corollary~\ref{cr:c4}:} 
Take $\epsilon=\frac{1-\gamma}{2}$ and the corresponding constant $D(\frac{1-\gamma}{2})$ in Corollary~\ref{cr:dh3}.
Let 
$$
j:= \min \left\{i\in\{0,1,2,\dots,m\} : H_{i,\max}^{(m)} \ge \frac{1+\gamma}{2} \right\}.
$$
Since $H_{i,\max}^{(m)}$ increases with $i$, we have $j\ge \theta_{\max}^{(m)}$. By Corollary~\ref{cr:dh3} we have
$$
0\le j - \theta_{\max}^{(m)} \le D(\frac{1-\gamma}{2}) +1.
$$
If $j>m-\Delta(m)$, then 
$$
\theta_{\min}^{(m)} - \theta_{\max}^{(m)}
\le m- \theta_{\max}^{(m)} \le m-j +D(\frac{1-\gamma}{2}) +1 < \Delta(m) +D(\frac{1-\gamma}{2}) +1 = o(\sqrt{m}),
$$
and by Theorem~\ref{thm:m4}, RM codes achieve capacity of $W$.
On the other hand, if $j\le m-\Delta(m)$, then by \eqref{eq:jj},
$$
 H_{j+\Delta(m),\min}^{(m)} \ge H_{j,\max}^{(m)} - \gamma
\ge \frac{1+\gamma}{2} - \gamma =  \frac{1-\gamma}{2}.
$$
Therefore according to Corollary~\ref{cr:dh3},
$$
j+\Delta(m) \ge \theta_{\min}^{(m)} - D(\frac{1-\gamma}{2}).
$$
Thus we obtain that
$$
\theta_{\min}^{(m)} - \theta_{\max}^{(m)} \le
\Big( j+\Delta(m) + D(\frac{1-\gamma}{2}) \Big) - \Big( j-D(\frac{1-\gamma}{2})-1 \Big)
= \Delta(m) + 2D(\frac{1-\gamma}{2}) +1 = o(\sqrt{m}),
$$
and by Theorem~\ref{thm:m4}, RM codes achieve capacity of $W$.
This completes the proof of Corollary~\ref{cr:c4}. \hfill \qedsymbol

\subsection{Proof of Theorem~\ref{thm:m5}} \label{sect:p5}
Our proof follows similar arguments to \cite{Weiss62,Allard72}.

Let $V_{m,r}$ be the generator matrix of Reed-Muller code $\cR(m,r)$, and let $k(m,r)$ be its dimension.
Let $U^{k(m,r)}$ be $k(m,r)$ i.i.d. Bernoulli-$1/2$ random variables and let $S^n$ be $n:=2^m$ i.i.d. Bernoulli-$p$ random variables, where $U^{k(m,r)}$ and $S^n$ are independent. Then Reed-Muller codes achieve capacity of BSC$(p)$ if and only if for each $m$ there is an $r$ such that
$$
\frac{k(m,r)}{n} = 1-h(p)-o(1)
$$
and that $U^{k(m,r)}$ can be decoded from $U^{k(m,r)} V_{m,r} + S^n$ with high probability.
Note that decoding $U^{k(m,r)}$ from $U^{k(m,r)} V_{m,r} + S^n$ is equivalent to decoding $S^n$ from $U^{k(m,r)} V_{m,r} + S^n$. Using the fact that the dual code of $\cR(m,r)$ is $\cR(m,m-r-1)$ \cite{Macwilliams77}, it is easy to verify that this is further equivalent to decoding $S^n$ from 
$$
(U^{k(m,r)} V_{m,r} + S^n)(V_{m,m-r-1})^T = S^n (V_{m,m-r-1})^T.
$$
Therefore Reed-Muller codes achieve capacity of BSC$(p)$ if and only if we can recover $S^n$ from $S^n (V_{m,m-r-1})^T$ with high probability, where the rate of compression is 
$$
\frac{n-k(m,r)}{n} = h(p)+o(1).
$$
This completes the proof of Theorem~\ref{thm:m5}.

\subsection{Proof of Proposition~\ref{prop:jww}} \label{sect:p6}
The cases of $m\le 2$ are trivial, so we only prove the cases of $m=3$ and $m=4$.
Let us start with $m=3$.
By Corollary~\ref{cr:od}, we only need to show that $H_{[3]}^{(3)}\ge H_{[2]}^{(3)}, H_{\{2,3\}}^{(3)} \ge H_{\{1\}}^{(3)}, H_{\{3\}}^{(3)}\ge H_{\emptyset}^{(3)}$.
In Section~\ref{sect:pst}, we already showed that $W_{[3]}^{(3)}$ and $W_{[2]}^{(3)}$ are the ``$-$" and ``$+$" polar transforms of $W_{[2]}^{(2)}$, respectively, and that $W_{\{3\}}^{(3)}$ and $W_{\emptyset}^{(3)}$ are the ``$-$" and ``$+$" polar transforms of $W_{\emptyset}^{(2)}$, respectively. Therefore, $H_{[3]}^{(3)}\ge H_{[2]}^{(3)}$ and $H_{\{3\}}^{(3)}\ge H_{\emptyset}^{(3)}$ follow immediately (and this extends to any dimension, i.e.,  the first and last transitions are always ordered due to polar codes).
Now let us prove 
\begin{equation}\label{eq:hre}
H_{\{2,3\}}^{(3)} \ge H_{\{1\}}^{(3)}
\end{equation}
using the equivalence between source and channel coding 
(see Section~\ref{sect:p5} for the discussion of the equivalence). Suppose that $X_1,X_2,\dots,X_8$ are i.i.d. Bernoulli-$p$ random variables, where $p$ is the crossover probability of the BSC channel.
Let $Y_1=\sum_{i=1}^8 X_i, Y_2=X_1+X_2+X_3+X_4, Y_3=X_1+X_2+X_5+X_6, Y_4=X_1+X_3+X_5+X_7, Y_5=X_1+X_2$. Then \eqref{eq:hre} is equivalent to
$H(Y_4|Y_1,Y_2,Y_3) \ge H(Y_5|Y_1,Y_2,Y_3,Y_4)$.
Notice that both $X_1$ and $X_2$ appear in $Y_1,Y_2,Y_3$. Therefore,
\begin{align*}
H(Y_4|Y_1,Y_2,Y_3) &= H(X_1+X_3+X_5+X_7|Y_1,Y_2,Y_3)
= H(X_2+X_3+X_5+X_7|Y_1,Y_2,Y_3) \\
& = H(Y_4+Y_5|Y_1,Y_2,Y_3) \ge H(Y_4+Y_5|Y_1,Y_2,Y_3,Y_4) =H(Y_5|Y_1,Y_2,Y_3,Y_4).
\end{align*}
This completes the proof of \eqref{eq:hre}.

For the case of $m=4$, again by Corollary~\ref{cr:od}, we only need to show that $H_{[4]}^{(4)}\ge H_{[3]}^{(4)}, H_{\{2,3,4\}}^{(4)} \ge H_{\{1,2\}}^{(4)},
H_{\{3,4\}}^{(4)} \ge H_{\{1\}}^{(4)}, H_{\{4\}}^{(4)}\ge H_{\emptyset}^{(4)}$.
In particular, $H_{[4]}^{(4)}\ge H_{[3]}^{(4)}$ and $H_{\{4\}}^{(4)}\ge H_{\emptyset}^{(4)}$ follow immediately from the discussions in Section~\ref{sect:pst}, so we only need to show the other two inequalities. We still use the equivalence between source and channel coding. Suppose that $X_1,X_2,\dots,X_{16}$ are i.i.d. Bernoulli-$p$ random variables, where $p$ is the crossover probability of the BSC channel.
Let 
\begin{align*}
Y_1 &= \sum_{i=1}^{16} X_i, \\
Y_2 &= X_1+X_2+X_3+X_4+X_5+X_6+X_7+X_8, \\
Y_3 &= X_1+X_2+X_3+X_4+X_9+X_{10}+X_{11}+X_{12}, \\
Y_4 &= X_1+X_2+X_5+X_6+X_9+X_{10}+X_{13}+X_{14}, \\
Y_5 &= X_1+X_3+X_5+X_7+X_9+X_{11}+X_{13}+X_{15}, \\
Y_6 &= X_1+X_2+X_3+X_4.
\end{align*}
Then $H_{\{3,4\}}^{(4)} \ge H_{\{1\}}^{(4)}$ is equivalent to $H(Y_5|Y_1,Y_2,Y_3,Y_4) \ge H(Y_6|Y_1,Y_2,Y_3,Y_4,Y_5)$.
Notice that both $X_1$ and $X_2$ appear in $Y_1,Y_2,Y_3,Y_4$. Therefore,
\begin{align*}
H(Y_5|Y_1,Y_2,Y_3,Y_4) &= H(X_1+X_3+X_5+X_7+X_9+X_{11}+X_{13}+X_{15}|Y_1,Y_2,Y_3,Y_4) \\
&= H(X_2+X_3+X_5+X_7+X_9+X_{11}+X_{13}+X_{15}|Y_1,Y_2,Y_3,Y_4).
\end{align*}
Similarly, both $X_3$ and $X_4$ appear in $Y_1,Y_2,Y_3$, and neither of them appears in $Y_4$. Therefore,
\begin{align*}
 & H(X_2+X_3+X_5+X_7+X_9+X_{11}+X_{13}+X_{15}|Y_1,Y_2,Y_3,Y_4)\\
= & H(X_2+X_4+X_5+X_7+X_9+X_{11}+X_{13}+X_{15}|Y_1,Y_2,Y_3,Y_4).
\end{align*}
Thus we conclude that 
\begin{align*}
H(Y_5|Y_1,Y_2,Y_3,Y_4) =  H(X_2+X_4+X_5+X_7+X_9+X_{11}+X_{13}+X_{15}|Y_1,Y_2,Y_3,Y_4) \\
= H(Y_5+Y_6|Y_1,Y_2,Y_3,Y_4) \ge = H(Y_5+Y_6|Y_1,Y_2,Y_3,Y_4,Y_5)
= H(Y_6|Y_1,Y_2,Y_3,Y_4,Y_5).
\end{align*}
This completes the proof of $H_{\{3,4\}}^{(4)} \ge H_{\{1\}}^{(4)}$. $H_{\{2,3,4\}}^{(4)} \ge H_{\{1,2\}}^{(4)}$ can be proved in the same way, and we omit its proof here.

\section{Simulation results}\label{sect:sim}
In this section we present some Monte Carlo simulation results over BEC channels.
More specifically, for a given code length $n=2^m$, we label all the subsets of $[m]$ as $A_1<A_2<\dots<A_n$ according to the order defined in Section~\ref{sect:intro}, and we use simulation to see how the conditional entropy $H_{A_i}^{(m,\BEC)}$ vary with $i$.
The number of iterations in our Monte Carlo simulation is $100000$.

We can see from Fig.~\ref{fig:vb1} and Fig.~\ref{fig:vb2} that $H_{A_i}^{(m,\BEC)}$ roughly decreases with $i$, meaning that RM codes also pick rows with the smallest conditional entropy, which is similar to polar codes.
We can also see a sharp transition of the conditional entropy from $1$ to $0$, which is necessary for RM codes to achieve capacity. Moreover, the transition becomes sharper for larger code length.

To compare the sharpness of transition with polar codes, we plot the sorted conditional entropy of polar codes in Fig.~\ref{fig:vb3}. We can see from Fig.~\ref{fig:vb2} and Fig.~\ref{fig:vb3} that for the same code length and channel erasure probability, the transition of RM codes is much sharper than that of polar codes. This is consistent with the properties of the fast polar transform detailed in Section \ref{sect:pst}. 

Finally, we remark that the condition \eqref{eq:nmv} in Corollary~\ref{cr:sv} is numerically verified on BEC for code length up to $256$ for various channel erasure probabilities.

\begin{figure}[h] 
\includegraphics[width=\textwidth]{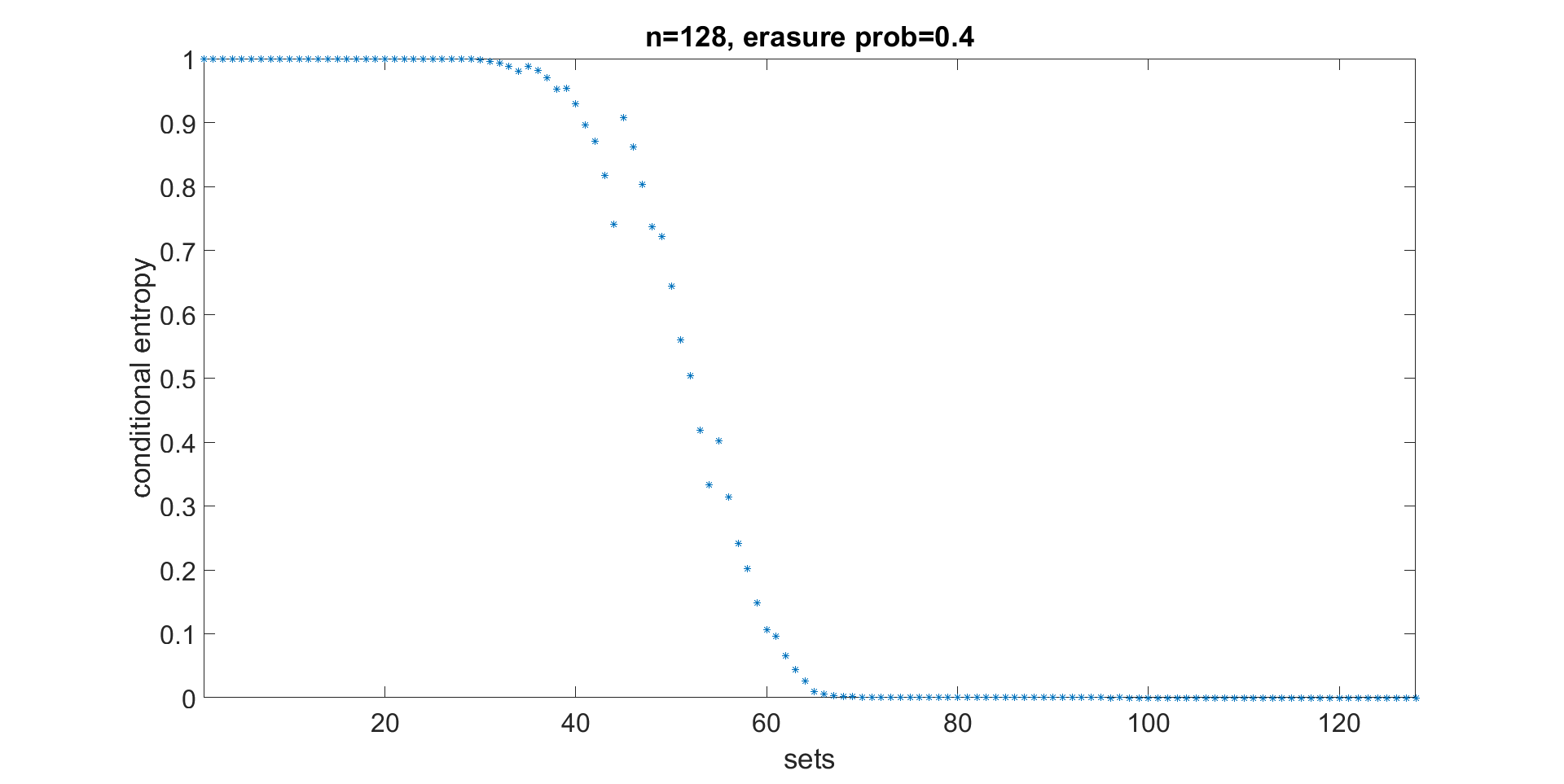}
\caption{Simulation results for Reed-Muller codes with length $n=128$. The channel erasure probability is set to be $0.4$. The $x$-axis corresponds to the index of the sets $i$ ranging from $1$ to $128$, and the $y$-axis corresponds to the conditional entropy $H_{A_i}^{(m,\BEC)}$.}
\label{fig:vb1}
\end{figure}

\begin{figure}[h] 
\includegraphics[width=\textwidth]{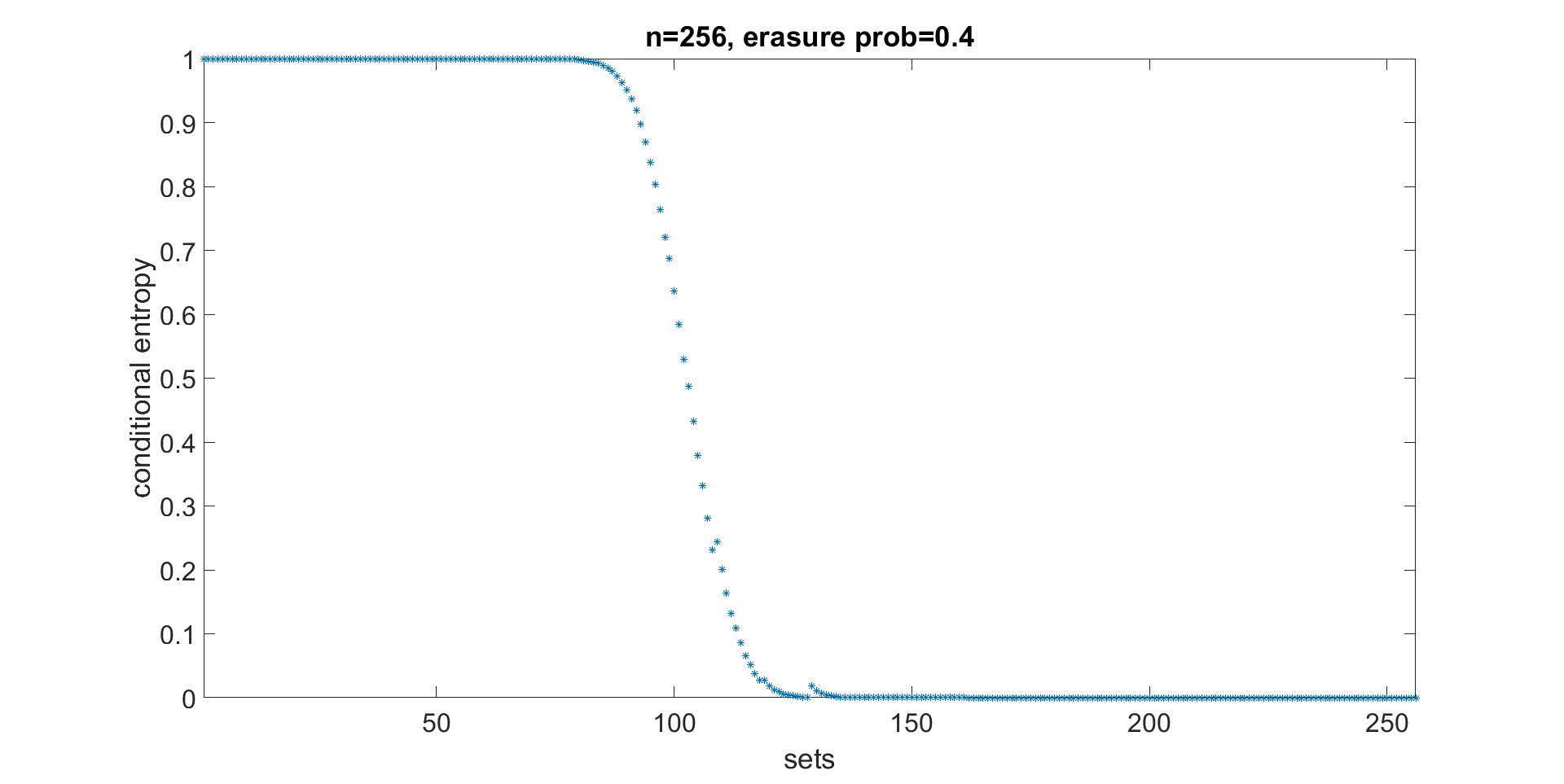}
\caption{Simulation results for Reed-Muller codes with length $n=256$. The channel erasure probability is set to be $0.4$. The $x$-axis corresponds to the index of the sets $i$ ranging from $1$ to $256$, and the $y$-axis corresponds to the conditional entropy $H_{A_i}^{(m,\BEC)}$.}
\label{fig:vb2}
\end{figure}

\begin{figure}[h] 
\includegraphics[width=\textwidth]{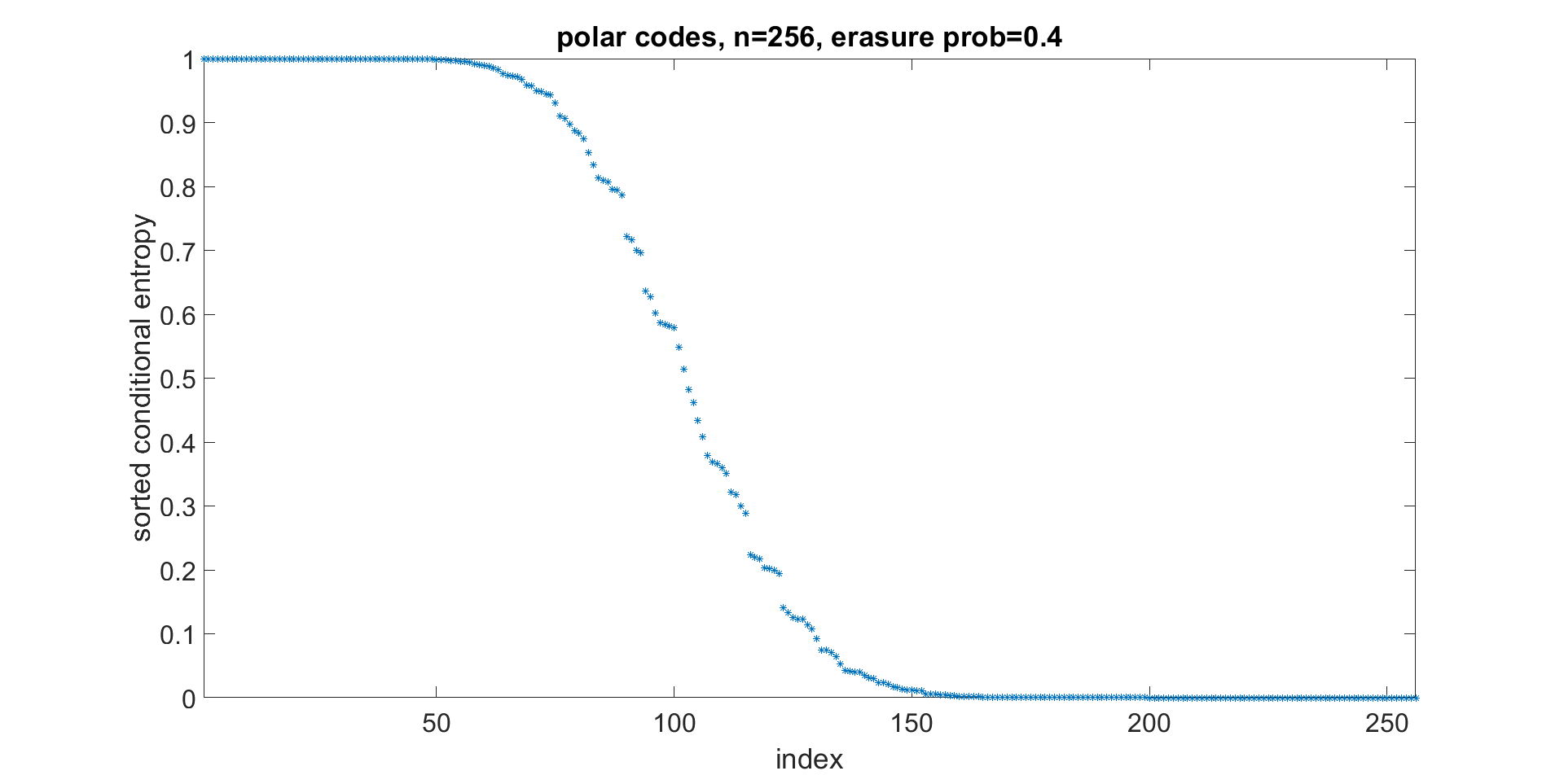}
\caption{Sorted conditional entropy of polar codes with length $n=256$. The channel erasure probability is set to be $0.4$. The transition from $1$ to $0$ is much slower than that of RM codes with the same length.}
\label{fig:vb3}
\end{figure}

\end{document}